\documentclass[runningheads]{llncs}

\usepackage[T1]{fontenc}

\newif\iftechreport\techreporttrue

\usepackage{graphicx}
\usepackage{appendix}
 
\usepackage{amsthm, amssymb, amsmath}
\usepackage{marvosym}
\usepackage{ebproof}
\usepackage{stmaryrd}
\usepackage{tikz}
\usepackage{caption}
\usetikzlibrary{automata, fit, arrows, positioning, graphs, matrix, calc}
\usepackage[nameinlink]{cleveref}

\newcommand{\chop}{%
  \mathchoice{\raisebox{5pt}{$\displaystyle\frown$}}
             {\raisebox{5pt}{$\frown$}}
             {\raisebox{5pt}{$\scriptstyle\frown$}}
             {\raisebox{5pt}{$\scriptscriptstyle\frown$}} }

\newcommand{\s}{\diamond}

\ExplSyntaxOn
\int_new:N \g__ebproof_sublevel_int
\box_new:N \g__ebproof_substack_box
\seq_new:N \g__ebproof_substack_seq

\cs_new:Nn \__ebproof_clear_substack:
  {
    \int_gset:Nn \g__ebproof_sublevel_int { 0 }
    \hbox_gset:Nn \g__ebproof_substack_box { }
    \seq_gclear:N \g__ebproof_substack_seq
  }

\cs_new:Nn \__ebproof_subpush:N
  {
    \int_gincr:N \g__ebproof_sublevel_int
    \hbox_gset:Nn \g__ebproof_substack_box
      { \hbox_unpack:N \g__ebproof_substack_box \box_use:c { \__ebproof_box:N #1 } }
    \seq_gput_left:Nv \g__ebproof_substack_seq
      { \__ebproof_marks:N #1 }
  }

\cs_new:Nn \__ebproof_subpop:N
  {
    \int_compare:nNnTF { \g__ebproof_sublevel_int } > { 0 }
      {
        \int_gdecr:N \g__ebproof_sublevel_int
        \hbox_gset:Nn \g__ebproof_substack_box {
          \hbox_unpack:N \g__ebproof_substack_box
          \box_gset_to_last:N \g_tmpa_box
        }
        \box_set_eq_drop:cN { \__ebproof_box:N #1 } \g_tmpa_box
        \seq_gpop_left:NN \g__ebproof_substack_seq \l_tmpa_tl
        \tl_set_eq:cN { \__ebproof_marks:N #1 } \l_tmpa_tl
      }
      { \PackageError{ebproof}{Missing~premiss~in~a~proof~tree}{} \__ebproof_clear:N #1 }
  }

\cs_new:Nn \__ebproof_append_subvertical:NN
  {
    \bool_if:NTF \l__ebproof_updown_bool
      { \__ebproof_append_above:NN #1 #2 }
      { \__ebproof_append_below:NN #1 #2 }
  }

\cs_new:Nn \__ebproof_join_subvertical:n
  {
    \group_begin:
    \__ebproof_subpop:N \l__ebproof_a_box
    \prg_replicate:nn { #1 - 1 }
      {
        \__ebproof_subpop:N \l__ebproof_b_box
        \__ebproof_enlarge_conclusion:NN \l__ebproof_b_box \l__ebproof_a_box

        \__ebproof_make_vertical:Nnnn \l__ebproof_c_box
          { \__ebproof_mark:Nn \l__ebproof_b_box {axis} - \__ebproof_mark:Nn \l__ebproof_b_box {left} }
          { \__ebproof_mark:Nn \l__ebproof_b_box {right} - \__ebproof_mark:Nn \l__ebproof_b_box {left} }
          { \skip_vertical:N \l__ebproof_rule_margin_dim }
        \__ebproof_vcenter:N \l__ebproof_b_box
        \__ebproof_append_subvertical:NN \l__ebproof_a_box \l__ebproof_c_box

        \__ebproof_append_subvertical:NN \l__ebproof_a_box \l__ebproof_b_box
      }
    \__ebproof_push:N \l__ebproof_a_box
    \group_end:
  }

\cs_new:Nn \__ebproof_renew_statement:nnn
  {
    \exp_args:Nc \RenewDocumentCommand { ebproof#1 }{ #2 } { #3 }
    \seq_gput_right:Nn \g__ebproof_statements_seq { #1 }
  }

\cs_generate_variant:Nn \clist_map_inline:nn { xn }
\__ebproof_renew_statement:nnn { infer } { O{} m O{} m }
  {
    \group_begin:
    \__ebproof_restore_statements:
    \keys_set_known:nnN { ebproof / rule~style } { #1 } \l_tmpa_tl
    \keys_set:nV { ebproof } \l_tmpa_tl
    \tl_set:Nn \l__ebproof_right_label_tl { #3 }

    \__ebproof_clear_substack:
    \clist_map_inline:xn { \clist_reverse:n { #2 } }
      {
        \__ebproof_join_horizontal:n { ##1 }
        \__ebproof_pop:N \l__ebproof_a_box
        \__ebproof_subpush:N \l__ebproof_a_box
      }
    \__ebproof_join_subvertical:n { \clist_count:n { #2 } }

    \__ebproof_push_statement:n { #4 }
    \__ebproof_join_vertical:
    \group_end:
  }
\ExplSyntaxOff

\begin{document}

\title{A Sequent Calculus For\\ Trace Formula Implication}

\author{Niklas Heidler\orcidID{0009-0001-9944-7587} \and
Reiner Hähnle\orcidID{0000-0001-8000-7613} }

\authorrunning{N. Heidler, R. H\"ahnle}

\institute{Technical University of Darmstadt, Germany\\
\email{<firstName>.<lastName>@tu-darmstadt.de}}

\maketitle     

\begin{abstract}
  Specification languages are essential in deductive program
  verification, but they are usually based on first-order logic, hence
  less expressive than the programs they specify. Recently,
  trace specification logics with fixed points that are at least as
  expressive as their target programs were proposed. This makes it
  possible to specify not merely pre- and postconditions, but the
  whole trace of even recursive programs.
  %
  % Trace formulas are a $\mu$-calculus consisting of unary
  % predicates, binary relations, a chop operator, and a least
  % fixed point operator~($\mu$) handling program recursion.
  Previous work established a sound and complete calculus to determine
  whether a program satisfies a given trace formula. However, the
  applicability of the calculus and its prospects for mechanized
  verification rely on the ability to prove consequence between trace
  formulas.
  We present a sound sequent calculus for proving implication (i.e.\
  trace inclusion) between trace formulas.  To handle fixed point
  operations with an unknown recursive bound, fixed point induction
  rules are used. We also employ contracts and $\mu$-formula
  synchronization.  While this does not yet result in a complete
  calculus for trace formula implication, it is possible to prove many
  non-trivial properties.

\end{abstract}

\section{Introduction}

There exist a variety of ways to specify and verify program properties
in a mechanized fashion. In \textit{Model Checking} \cite{MC},
temporal logic, such as Linear Temporal Logic (LTL) or Computation
Tree Logic (CTL) is used to specify program behaviour. During
verification, a model of the given program and its temporal logic
specification are finitely unwound, typically by automata
constructions. \textit{Deductive Verification} \cite{Hähnle2019} uses
first-order logic (FOL) to formalize procedure contracts in Hoare
calculus \cite{Hoare69} or in program logic \cite{KeYTutorial24} to
prove that a given first-order postcondition holds in any state
reachable by executing the given procedure, assuming that a
precondition held in the start state.

It is interesting to note that ---with few exceptions
\cite{McGuireMW94,SprengerDam03b}--- specification languages in
deductive verification are weaker in expressiveness than the programs
they are supposed to specify. Moreover, nearly all deductive
verification techniques are based on reasoning about intermediate
states, i.e. before and after a procedure call. In this sense, model
checking is more natural, because there is a direct correspondence
between the program model and its specification. However, LTL and CTL,
certain extensions \cite{AEM04} notwithstanding, cannot express
modular verification over contracts and they target \emph{models} of
programs. In consequence, an obvious question arises: Is there a logic
that permits trace-based \emph{and} contract-based specification of
imperative programs with recursive procedures that has a natural
correspondence between program and specification?

This was recently answered affirmatively in the form of a trace
specification logic with smallest fixed points. Here, trace formulas
$\Phi$ specify a (possibly infinite) set of finite computation traces
generated by a program $S$ from a simple imperative language
\texttt{Rec} with recursive procedure declarations. Judgments take the
form $S{\,:\,}\Phi$ and mean: Any possible execution trace of $S$ is
contained in the set of traces characterized by $\Phi$.  Gurov \&
Hähnle \cite{GurovHaehnle24} provide a sound, complete, and
\emph{compositional} proof calculus for judgments of the form
$S{\,:\,}\Phi$, where ``compositional'' means that the rule premises
do not introduce intermediate formulas not present in the
conclusion. However \emph{weakening} of trace formulas (i.e.~prove
$S{\,:\,}\Psi$ instead of $S{\,:\,}\Phi$ provided that $\Psi$ implies
$\Phi$) is still necessary.

Soundness and completeness of the calculus rest on a strong
correspondence between programs and trace formulas: For any
\texttt{Rec} program $S$, there exists a \textit{strongest trace
  formula} $\mathit{stf}(S)$ that characterizes \emph{exactly} the traces
generated by $S$.\footnote{The paper \cite{GurovHaehnle24} even proves
  the reverse direction: For any trace formula $\Phi$ there is a
  \emph{canonical program} $S$ having exactly the same traces as
  $\Phi$, establishing a Galois connection between programs and trace
  formulas. However, this result is not relevant for the present
  paper.}  Hence, $S{\,:\,}\Phi$ is valid iff the traces specified by
$\mathit{stf}(S)$ are included in the traces specified by $\Phi$. This implies
one can verify a judgment $S{\,:\,}\Phi$ by simply proving the trace
formula consequence $\mathit{stf}(S) \models \Phi$. Alternatively, one can use
the rules of the calculus to prove $S{\,:\,}\Phi$ directly.  Thus, the
correspondence between programs and trace formulas creates the
opportunity to verify judgments with a program calculus \emph{or} by
trace formula consequence. It is also possible to mix both styles, of
course. In either case, weakening is needed for completeness, so
implication between trace formulas is a crucial ingredient.  This
requires a separate proof system and such a calculus was considered as
future work in \cite{GurovHaehnle24}. It is the main objective of the
present paper.

The consequence relation between formulas in a fixed point logic is a
difficult problem---because trace formulas are as expressive as
recursive programs it is highly undecidable. Therefore, our
investigation into how far one can get with such a calculus, is
interesting in its own right. Existing literature has little to say
about the topic.
%
% Based on the concept of strongest trace formulas, Hähnle et
% al. \cite{lipics} additionally prove that for every program $S$ the
% judgement $S{\,:\,}stf(S)$ is indeed derivable.
% However, the applicability of the calculus relies on the ability to
% handle the implication between trace formulas, which has not been
% investigated yet. The same can be said for various other calculi
% handling trace-based verification, such as \cite{hoho}.
%
% However, establishing this kind of trace inclusion proof is
% non-trivial and needs a specialized formal calculus to derive
% corresponding instances. Even if this calculus cannot achieve
% completeness, it can also be used to strengthen formulas in
% judgements $S{\,:\,}\Phi$ to $S{\,:\,}\Phi'$, as long as we could
% establish the trace consequence $\Phi' \models \Phi$, opening up
% another use case for a trace implication calculus.
%
% , which is able to derive tuples of trace formulas $(\Phi, \Psi)$ with
% $\Phi \models \Psi$.
%
% , specifically focusing on deriving trace inclusions of the shape
% $stf(S) \models \Psi$.
%%%%%%%%%%%%%%%%%%%%%%%%%%%%%%%%%%%%%%%%%%%%%%%%%%%%%%%%%%%%%
%%%%%%%%%%%%%%%%%%%%%%%%%%%%%%%%%%%%%%%%%%%%%%%%%%%%%%%%%%%%%
%%%%%%%%%%%%%%%%%%%%%%%%%%%%%%%%%%%%%%%%%%%%%%%%%%%%%%%%%%%%%
%
The central challenge in the design of a calculus for implication of
trace formulas is the handling of fixed point formulas, i.e.\ formulas
with a leading fixed point operator $\mu$. We propose increasingly
complex strategies of how to eliminate fixed point formulas, without
reaching completeness yet:
\begin{enumerate}
\item Straightforward \emph{unfolding} of $\mu$-formulas is sufficient
  to deal with executions that have concrete bounds
  (Section~\ref{sec:base-rules}).
\item Fixed point \emph{induction} lets one prove trace inclusion of
  recursive executions with an unknown (or very high) bound
  (Section~\ref{sec:fixed point-induction}).
\item To capture the execution state \emph{after} a fixed point
  formula we equip the calculus with Hoare-style state-based procedure
  contracts. The logic and calculus is expressive enough to prove such
  contracts and to propagate them inside the proofs, without the need
  to refer to meta theorems (Section~\ref{sec:contracts}).
\item When proving the consequence relation between two
  $\mu$-formulas, one often encounters the problem that the execution
  of their bodies is not \emph{synchronized}. We equip the calculus
  with $\mu$-formula synchronization rules
  (Section~\ref{sec:synchronization}) that are able to synchronize
  recursive variables inside fixed point operations in many, but not
  in all cases. This is one source of incompleteness.
\end{enumerate}

%
% Since fixpoint operations typically lose state information that may
% be necessary for deriving the continuation of the trace, we will
% extend the calculus with Hoare-style method contracts. These method
% contracts will be proven valid via our calculus and then applied in
% order to propagate necessary state information over fixpoint
% operations. This ensures derivability of trace formulas in the
% antecedent, in which recursive variables occur non-tailrecursive and
% method calls occur mid-program.
%
% Whilst we will prove that our calculus rules are all sound, we will
% \textit{not achieve completeness}, as is indicated by various open
% incompleteness problems. Extending this calculus with the goal of
% achieving - and possibly even proving - completeness, is hence left
% for future~research.
%
The paper is structured as follows: In \Cref{ch2}, we introduce
\texttt{Rec} programs. Trace formulas are defined in \Cref{ch3},
together with some examples of provable properties. \Cref{ch4}
proposes a basic calculus for trace implication, which is the core of
this paper. \Cref{ch5} extends the basic calculus with method
contracts and $\mu$-formula synchronization. \Cref{ch6} refers to
related work, while \Cref{ch7} concludes the paper and proposes future
work.  As noted, completeness is elusive at the moment, however, we
are able to prove a range of interesting and non-trivial properties.

%%% Local Variables:
%%% mode: latex
%%% TeX-master: "main"
%%% End:
  
\section{The Rec Language}
\label{ch2}

We define a simple imperative programming language \texttt{Rec}
\cite{GurovHaehnle24} with (recursive) procedure calls.

\begin{definition}[Rec Program] A \textit{\texttt{Rec} Program} is a pair
  $(S, T)$, where $S$ is a \textit{\texttt{Rec} Statement} generated
  by the grammar
  {\normalfont
    $$
    S ::= \text{\textbf{skip} | } x := a \text{ | } S;S \text{ | }
    \text{\textbf{if} } b \text{ \textbf{then} } S \text{
      \textbf{else} } S \text{ | } m()
    $$}%
  \noindent and $T$ is a possibly empty sequence $M^*$ of procedure
  declarations, where each $M$ declares a parameter-less procedure
  $M \equiv m \{S\}$ consisting of \textit{procedure name} $m$ and
  \textit{procedure body} $S$. Schema variables $a$ and $b$ range over
  side-effect free arithmetic and boolean expressions, respectively,
  that are not further specified.
\end{definition}

The \texttt{Rec} language does not include syntax for
\textbf{while}-loops, however, these can easily be modeled with the
help of a tail-recursive procedure. There is also no parameter passing
mechanism.

A program \emph{trace} $\sigma$ is a, possibly empty, finite sequence
of execution \emph{states}~$s$, partial mappings from program
variables $x$ to integer values. Regarding the semantics of a program
in terms of its finite $traces(S)$ of statements~$S$, we refer to the
standard definitions in the literature \cite{GurovHaehnle24}.

\begin{example} The factorial \texttt{Rec} Program $(S_{fac},\,T_{fac})$ is given by the statement
  $S_{fac} \equiv y := 1;\, factorial()$ and the procedure table
  $$
  T_{fac} \equiv factorial \{ \text{\textbf{if} } x = 1 \text{
    \textbf{then} } skip \text{ \textbf{else} } y := y * x;\, x := x -
  1;\, factorial() \}
  $$

  By convention, sequential composition binds stronger that the
  conditional, i.e.\ the final three statements form the \textbf{else}
  block.  For any start state $s=[x\mapsto i]$ with $i>0$, the program
  computes the factorial of $x$ and stores the result in $y$, i.e.\
  the program terminates in a state $s'$ where $s'(y)=x!$.
\end{example}

%%% Local Variables:
%%% mode: latex
%%% TeX-master: "main"
%%% End:

\section{Trace Formulas}
\label{ch3}

We define the \emph{trace formula logic}. Like for \texttt{Rec}
programs, the semantics of its formulas is given as a set of program
traces.

\begin{definition}[Trace Formula Syntax] The grammar of \emph{trace
    formulas} is
  {\normalfont
    $$
    \Phi ::= p \text{ | } R % \text{ | } true
    \text{ | } \Phi \land \Phi \text{ | }
    \Phi \lor \Phi \text{ | } \Phi \chop \Phi \text{ | } X \text{ | }
    \mu X. \Phi
    $$}%
  where $p$ ranges over first-order state predicates $Pred$, $R$
  ranges over binary relations between states, and $X$ ranges over
  recursion variables \textit{RVar}. The binary operator $\chop$ is called
  chop.\footnote{It is inspired by Interval Temporal Logic
    \cite{HalpernMM83} and its use in specification by
    \cite{nak-uus-09}.} We assume $R$ contains at least the relations
  \begin{center}{\normalfont $Id := \{ (s, s) \in State^2\}$} and
    {\normalfont
      $Sb_x^a := \{ (s, s') \in State^2 \text{ | } s' = s[x \mapsto
      \mathbb{A}\llbracket a\rrbracket(s)] \}$}\enspace.
  \end{center}
\end{definition}

Relation $Id$ models a skip and $Sb_x^a$ an
assignment. $\mathbb{A}\llbracket a\rrbracket(s)$ refers to the
evaluation of arithmetic expression $a$ in state $s$. Observe that the
logic is not closed under negation: only smallest fixed point formulas
are permitted.

\begin{definition}[Trace Formula Semantics] Each trace formula $\Phi$
  evaluates to a set of finite traces. Given a \textit{valuation
    function} $\mathbb{V} : \mathit{RVar} \rightarrow P(State^+)$ that maps
  recursion variables to sets of traces, the \textit{semantics} of a
  trace formula $\Phi$ under valuation $\mathbb{V}$, denoted
  $\llbracket \Phi\rrbracket_{\mathbb{V}}$, is defined by the
  equations in~\Cref{fig:semantics-formulas}.
  $\llbracket \Phi\rrbracket$ abbreviates
  $\llbracket \Phi\rrbracket_{\mathbb{V}}$ when $\mathbb{V}$ does not
  affect the result.
\end{definition}

\begin{figure}[t]
  \begin{align*}
    \llbracket p\rrbracket_{\mathbb{V}} & = \{ s \cdot \sigma \hspace{-0.5mm} \text{ | } \hspace{-0.5mm} s \models p \land \sigma \in State^* \} & \llbracket R\rrbracket_{\mathbb{V}} &= \{ s \cdot s' \text{ | } R(s, s') \} \\
    \llbracket \Phi_1 \land \Phi_2\rrbracket_{\mathbb{V}} &= \llbracket \Phi_1\rrbracket_{\mathbb{V}} \cap \llbracket\Phi_2\rrbracket_{\mathbb{V}} & \llbracket \Phi_1 \lor \Phi_2\rrbracket_{\mathbb{V}} &= \llbracket \Phi_1\rrbracket_{\mathbb{V}} \cup \llbracket\Phi_2\rrbracket_{\mathbb{V}} \\
    \llbracket \Phi_1 \chop \Phi_2\rrbracket_{\mathbb{V}} & = \{\sigma \cdot s \cdot \sigma' \hspace{1mm} | \hspace{1mm} \sigma \cdot s \in \llbracket \Phi_1\rrbracket_{\mathbb{V}} \land s \cdot \sigma' \in \llbracket \Phi_2\rrbracket_{\mathbb{V}} \} & \llbracket X\rrbracket_{\mathbb{V}} &= \mathbb{V}(X) \\
    \llbracket \mu X. \Phi\rrbracket_{\mathbb{V}} &= \bigcap\{ \gamma \subseteq State^+ \text{ | } \llbracket \Phi\rrbracket_{\mathbb{V}[X \mapsto \gamma]} \subseteq \gamma \} % & \llbracket true\rrbracket_{\mathbb{V}} &= State^+
  \end{align*}
  \begin{center}\vspace*{-3em}
  \end{center}
  \caption{Semantics of trace formulas}
  \label{fig:semantics-formulas}
\end{figure}

\noindent Observe that $\llbracket \mu X. \Phi\rrbracket_{\mathbb{V}}$
maps to the least fixed point of $\Phi$ in the powerset lattice
$(P(State^+), \subseteq)$. This is justified by monotonicity of
$\lambda \gamma. \llbracket \Phi\rrbracket_{\mathbb{V}[X \mapsto
  \gamma]}$ and the Knaster-Tarski theorem.

\begin{theorem}[Strongest Trace Formula~\cite{GurovHaehnle24}]
  \label{th1}
  For each \textit{\texttt{Rec} Program} $(S, T)$ there exists a
  \textit{closed strongest trace formula} $\Phi$ with
  $traces(S) = \llbracket\Phi\rrbracket$.
\end{theorem}

The strongest trace formula can be effectively constructed from a
given \texttt{Rec} program. The details of the construction and the
proof are in \cite{GurovHaehnle24}. The theorem implies that trace
formulas are at least as expressive as the \texttt{Rec} language.

\begin{example}\label{ex:fac}
  Trace formula $\Phi_{fac}$ is the strongest trace formula for
  $(S_{fac}, T_{fac})$:
  $$
  \Phi_{fac} \equiv Sb_y^1 \chop Id \chop \Phi_{m}\text{, where}$$
  $$
  \Phi_{m} \equiv \mu X_{fac.} ((x = 1 \land Id \chop Id) \lor (x \neq 1 \land Id \chop Sb_y^{y*x} \chop Sb_x^{x-1} \chop Id \chop X_{fac}))$$
\end{example}

\begin{definition}[Satisfiability]
  A \texttt{Rec} program $S$ satisfies a trace formula $\Phi$ (write
  $S{\,:\,}\Phi$) iff $traces(S) \subseteq \llbracket\Phi\rrbracket$.
\end{definition}

As noted in the introduction, a sound and complete compositional proof
calculus for $S{\,:\,}\Phi$ is given in \cite{GurovHaehnle24}, but its
applicability relies on weakening, i.e.~the semantic entailment oracle
$\Phi\models\Psi$, of which this paper presents the first formal
investigation.

Theorem~\ref{th1} implies that trace formulas can characterize the
traces of a given \texttt{Rec} program \emph{precisely}. But this does
not mean that trace formulas are just an alternative notation for
programs: Unlike programs, with the help of suitable binary
predicates, formulas can conveniently abstract away from computational
details.

\begin{example}
  Let $S_{down}$ be a \texttt{Rec} program that decreases a variable
  $x$ by $2$ until $x$ reaches the value $0$. Afterwards, it further
  decreases variable $x$ by $1$.  Whether the recursion is entered
  depends on the initial value of $x$.
  \begin{align*}
    S_{down}  \equiv\,  &  down()  \text{  with  }\\
                        &  down  \{\text{\textbf{if}  }  x  =  0  \text{  \textbf{then}  }  x  :=  x  -  1  \text{  \textbf{else}  }  x  :=  x  -  2;  down()\}
  \end{align*}

  We illustrate how to use trace formulas as an abstract specification
  $S_{down}$ with two properties. None of them can be expressed with
  Hoare-style contracts based on pre- and postconditions.
  \begin{enumerate}
  \item For variable $x$, define the relation
    $R_{dec}^x := \{ (s, s') \in State^2\mid s(x)\geq s'(x)\}$ which
    is easy to axiomatize.  The property that $x$ never increases
    throughout the execution of program $S_{down} $ is expressed with
    the fixed point formula
    $ \mu X_{dec.} R_{dec}^x \lor R_{dec}^x \chop X_{dec} $. 
\item If $x$ is \textit{even} and \textit{non-negative}, then $x$ will
  eventually reach value $0$. Afterwards, $x$ will eventually reach
  value $-1$:
  $$
  \overline{even(x)} \lor x < 0 \lor true \chop x = 0 \chop x = -1\enspace.
  $$

  Observe that the negated expressions $\overline{even(x)}$ and $x<0$
  serve to impose their complement as a constraint on the initial
  state of a trace. Trace formulas are not closed under negation, but
  Boolean expressions related to individual states are.  Also observe
  that the semantics of atomic trace formulas $p$ implies that between
  the states with $x=0$ and $x=-1$ an arbitrary number of intermediate
  states can occur.
\end{enumerate}

\end{example}

The properties above were proven as judgments in the calculus provided
in \cite{GurovHaehnle24}, while the necessary weakening steps were
proven in the calculus presented in Section~\ref{ch4}.  The
derivations can be found in \cite{Heidler24}.

%%% Local Variables:
%%% mode: latex
%%% TeX-master: "main"
%%% End:

\section{A Proof Calculus for Trace Formula Consequence}
\label{ch4}

% Determining the validity of trace formula consequences
% $\Phi \models \Psi$ (i.e.
% $\llbracket\Phi\rrbracket \subseteq \llbracket\Psi\rrbracket$) is
% imperative for many calculi focusing on trace-based verification,
% such as \cite{GurovHaehnle24} or~\cite{hoho}, and at least as
% difficult as program implication. We will now provide a sound, but
% incomplete, formal calculus for this problem.

\subsection{Sequents}

\begin{definition}[Sequents]
  A \emph{sequent} in our calculus has the shape
  $\xi \s \Gamma \vdash \Delta$, where
  $\xi \subseteq \mathit{RVar} \times Pred \times \mathit{RVar}$ and $\Gamma, \Delta$
  are sets of trace formulas. A triple $(X, p, X') \in \xi$ is written
  $(X|p, X')$ as syntactic sugar.
\end{definition}

The purpose of $\xi$ is to specify constraints on the recursion
variables occurring in a valuation $\mathbb{V}$. We write
$\Gamma \vdash \Delta$ as an abbreviation for
$\emptyset \s \Gamma \vdash \Delta$ in case $\xi$ is~empty or irrelevant. $\xi$ is
always empty for a top-level sequent.

\begin{definition}[Validity of Sequents]
  A sequent $\xi \s \Gamma \vdash \Delta$ is \emph{valid}, if for all
  valuations $\mathbb{V}$ with
  $\llbracket X \land p\rrbracket_{\mathbb{V}} \subseteq \llbracket
  X'\rrbracket_{\mathbb{V}}$ for all $(X|p, X') \in \xi$, it is the
  case that
  $\llbracket \bigwedge\Gamma\rrbracket_{\mathbb{V}} \subseteq
  \llbracket\bigvee\Delta\rrbracket_{\mathbb{V}}$.
\end{definition}

\begin{example}
  Let $X_1$ and $X_2$ be recursion variables. Then
  $$
  (X_1|_{x \geq 0}, X_2) \s x = 0, X_1 \vdash X_2
  $$
  is a (trivially) valid sequent, because $(X_1|_{x \geq 0}, X_2)$ already assumes trace inclusion between $X_1$ and $X_2$, whenever $x \geq 0$.
\end{example}

\subsection{Base Rules}
\label{sec:base-rules}

\begin{definition}[Program State]
  To extract the current state from the antecedent $\Gamma$ of a
  sequent, we define
  $P_\Gamma := \{ p \in \Gamma {\normalfont \text{ | }} p \in Pred \}$
  as the set of all first-order state predicates occurring in
  $\Gamma$.
\end{definition}

\paragraph{First-order Rules.}  Standard axioms such as \texttt{CLOSE},
\texttt{TRUE} and \texttt{FALSE}, as well as the usual rules of the
first-order sequent calculus are not separately listed. They are all
valid in our setting.

\begin{figure}[t]
  \begin{align*} 
    &\begin{prooftree} 
      \hypo{\xi \s \Gamma, p \vdash \Delta}
      \hypo{\xi \s \Gamma, \overline{p} \vdash \Delta}
      \infer[left label=\texttt{CUT}]2{\xi \s \Gamma \vdash \Delta}
    \end{prooftree} \hspace{5mm}
       \begin{prooftree}
      \infer[left label=\texttt{REL}]0[$\underbrace{\{(s, s') \in R \hspace{1mm} | \hspace{1mm} s \models P_\Gamma\}}_{R|_{P_\Gamma}} \subseteq R'$]{\xi \s \Gamma, R \vdash R', \Delta}
\end{prooftree} \\[7pt]
    &\begin{prooftree}
        \hypo{P_\Gamma \vdash q}
        \hypo{\xi \s \Gamma, q \vdash \Delta}
        \infer[left label=\texttt{PRED}]2{\xi \s \Gamma \vdash \Delta}
      \end{prooftree} \hspace{10mm}
      \begin{prooftree}
        \hypo{P_\Gamma \vdash p}
        \infer[left label=\texttt{RVAR}]1{\xi, (X_1|_{p}, X_2) \s \Gamma, X_1 \vdash X_2, \Delta}
      \end{prooftree}\\[10pt]
    &\begin{prooftree}
      \hypo{\xi \s P_\Gamma, Id \vdash \Psi_1}
      \hypo{\cdots}
      \hypo{\xi \s P_\Gamma, Id \vdash \Psi_n}
      \hypo{\xi \s P_\Gamma, \Phi \vdash \Psi_1',\ldots, \Psi_n'}
      \infer[left label=\texttt{CH-ID}]4{\xi \s \Gamma, Id \chop \Phi \vdash \Psi_1 \chop \Psi_1',\ldots, \Psi_n \chop \Psi_n', \Delta}
    \end{prooftree} \\[7pt]
    &\begin{prooftree}
      \hypo{\xi \s P_\Gamma, Sb_x^a \vdash \Psi_1}
      \hypo{\cdots}
      \hypo{\xi \s P_\Gamma, Sb_x^a \vdash \Psi_n}
      \hypo{\xi \s spc_{x := a}(P_\Gamma), \Phi \vdash \Psi_1', \ldots, \Psi_n'}
      \infer[left label=\texttt{CH-UPD}]4{\xi \s \Gamma, Sb_x^a \chop \Phi \vdash \Psi_1 \chop \Psi_1', \ldots, \Psi_n \chop \Psi_n', \Delta}
    \end{prooftree}
  \end{align*}
  \begin{center}\vspace*{-2em}
  \end{center}
  \caption{Calculus rules for predicates and relations}
  \label{fig:Base}
\end{figure}

\paragraph{Rules for Predicates and Binary Relations (\Cref{fig:Base}).}

The rule \texttt{CUT} performs a case distinction on predicate $p$. In
contrast to trace formulas, first-order formulas are closed under
negation. Rule \texttt{PRED} infers information from the program state
in its first premise and adds it to the antecedent of its second
premise.

Axiom \texttt{REL} handles trace inclusion between binary relations. Observe that the current
program state $P_\Gamma$ further restricts relation $R$ in the antecedent, abbreviated as $R|_{P_\Gamma}$. Rule~\texttt{RVAR} characterizes trace inclusion between recursion
variables based on $\xi$, and needs to prove the corresponding
restricting predicate in its premise.

Rules \texttt{CH-ID} and \texttt{CH-UPD} handle the case where a
binary relation occurs at the beginning of the current chop sequence
in the antecedent. In both rules, the first $n$ premises ensure that
the leading relation of the antecedent infers the leading formulas of
corresponding chop operations in the succedent. The inference between
the remaining trace formula composites occurs in the final premise.
As the leading binary relation in the antecedent may change program
variables, the program state may need to be adapted to reflect those
changes. For this reason we restrict ourselves to relations $Id$ and
$Sb_x^a$ in the antecedent which is sufficient to define strongest
trace formulas (the rules can be easily extended to support other
binary relations in the antecedent).
The program state for the remaining trace is preserved when the
leading relation is $Id$. In case of $Sb_x^a$, however, the program
state $P_\Gamma$ needs to be updated to its strongest
postcondition~\cite{dijkstra1976discipline} relative to state update
$x := a$, indicated by $spc_{x := a}(P_\Gamma)$.

\begin{figure}[t]
  \begin{prooftree}
    \infer[left label=\texttt{REL}]0[\hspace{-1mm}$Sb_y^{y*x}|_{\mathbf{P_\Gamma^2}} \hspace{-0.5mm} \subseteq \hspace{-0.5mm} R_{inc}^y$]{P_\Gamma^2, Sb_y^{y*x} \vdash R_{inc}^y}
    \hypo{}
    \ellipsis{}{P_\Gamma^4 \vdash \bigwedge P_\Gamma^1}
    \infer[left label=\texttt{RVAR}]1{(X_{1}|_{\bigwedge P_\Gamma^1}, X_{2}) \s P_\Gamma^4, X_{1} \vdash X_{2}}
    \ellipsis{}{\hspace{-4mm}(X_{1}|_{\bigwedge P_\Gamma^1}, X_{2}) \s P_\Gamma^3, Sb_x^{x-1} \chop X_{1} \vdash R_{inc}^y \chop X_{2}}
    \infer[left label = \texttt{CH-UPD}]2{(X_{1}|_{\bigwedge P_\Gamma^1}, X_{2}) \s P_\Gamma^2, Sb_y^{y*x} \chop Sb_x^{x-1} \chop X_{1} \vdash R_{inc}^y \chop R_{inc}^y \chop X_{2}}
  \end{prooftree}
  \begin{center}\vspace*{-2em}
  \end{center}
  \caption{Demonstration of predicate and relation rules}
  \label{fig:dem1}
\end{figure}

\begin{example}\label{ex:states}
  Consider the following four state predicates
  $P_\Gamma^1 \equiv \{x \geq 1, y \geq 1\}$, $P_\Gamma^2 \equiv \{x > 1, y \geq 1\}$, $P_\Gamma^3 \equiv \{x > 1, y \geq x\}$ and $P_\Gamma^4 \equiv \{x \geq 1, y > x\}$, and define a new binary
  relation $R_{inc}^y := \{(s, s') \text{ | } s(y) \leq s'(y)\}$,
  expressing that program variable $y$ does not decrease. An example
  derivation is in \Cref{fig:dem1}. It proves that for the constraints
  on valuations expressed in $P_\Gamma^1,\,P_\Gamma^2,\,P_\Gamma^3,\,P_\Gamma^4$, the sequence
  of state updates $y:=y*x;\,x:=x-1$ can be approximated by
  non-decreasing predicates of program variable $y$. 
  
  % Keep in mind that the strongest postcondition of $P_\Gamma^1$ relative to state update $\mathit{y:=y*x}$ is (strictly speaking) \textit{not} $P_\Gamma^1$. $P_\Gamma^1$ is, however, a weaker version that is still strong enough to close the proof, chosen due to readability purposes. The strongest postcondition of $P_\Gamma^1$ relative to $x:=x-1$ is $P_\Gamma^2$.
\end{example}

\paragraph{Rules for Unfolding and Lengthening (\Cref{fig:Unf}).}

The rules \texttt{UNFL} and \texttt{UNFR} \emph{unfold} a fixed
point formula $\Phi$ in the antecedent and succedent,
respectively. This is sound, because $\mu X. \Phi$ is the least fixed
point, implying that an additional recursive application does not
change its semantic evaluation.

\begin{figure}[t]
  \begin{align*} 
    &\begin{prooftree}
      \hypo{\xi \s \Gamma, \Phi[\mu X. \Phi/X] \vdash \Delta}
      \infer[left label=\texttt{UNFL}]1{\xi \s \Gamma, \mu X. \Phi \vdash \Delta}
    \end{prooftree} \qquad \qquad \hspace{13mm} 
      \begin{prooftree}
        \hypo{\xi \s \Gamma \vdash \Psi[\mu X. \Psi/X], \Delta}
        \infer[left label=\texttt{UNFR}]1{\xi \s \Gamma \vdash \mu X. \Psi, \Delta}
      \end{prooftree} \\[7pt]
    &\begin{prooftree}
      \hypo{\xi \s \Gamma, \mu X. repeat_i(\Phi) \vdash \Delta}
      \infer[left label=\texttt{LENL}]1[$i \geq 1$]{\xi \s \Gamma, \mu X. \Phi \vdash \Delta}
    \end{prooftree} \qquad \qquad
      \begin{prooftree}
        \hypo{\xi \s \Gamma \vdash \mu X. repeat_i(\Psi), \Delta}
        \infer[left label=\texttt{LENR}]1[$i \geq 1$]{\xi \s \Gamma \vdash \mu X. \Psi, \Delta}
      \end{prooftree}
  \end{align*}
  \begin{center}\vspace*{-2em}
  \end{center}
  \caption{Calculus rules for unfoldings and lengthenings}
  \label{fig:Unf}
\end{figure}

Rules \texttt{LENL} and \texttt{LENR} \emph{lengthen} fixed point
formula $\Phi$ in the antecedent and succedent respectively. The
repetition of fixed point formulas can be defined as
$$
repeat_0(\Phi) := \Phi \text{ and } repeat_i(\Phi) :=
\Phi[repeat_{i-1}(\Phi)/X]) \text{ for } i \geq 1\enspace.
$$
The rules are sound, because for any recursive procedure $m$,
procedure $m$ with $n$ recursive calls inlined has the same least
fixed point as $m$ itself.

\begin{example}
  Let $\Phi \equiv \mu X. (R \lor R \chop X)$ be the fixed point
  formula modeling transitive closure of a binary relation $R$. Its unfolding is $R \lor R \chop \Phi$, while lengthening it once corresponds to $\mu X. (R \lor R \chop (R \lor R \chop X))$.
\end{example}

\begin{figure}[t]
  \begin{align*} 
    &\begin{prooftree}
      \hypo{\xi \s \Gamma \vdash \Psi, \Delta}
      \infer[left label=\texttt{ARB1}]1{\xi \s \Gamma \vdash true \chop \Psi, \Delta}
    \end{prooftree} \qquad \qquad
      \begin{prooftree}
        \hypo{\xi \s \Gamma, \Phi_1 \chop \Phi_2 \vdash \Phi_1 \chop true \chop \Psi, \Delta}
        \infer[left label=\texttt{ARB2}]1{\xi \s \Gamma, \Phi_1 \chop \Phi_2 \vdash true \chop \Psi, \Delta}
      \end{prooftree}
  \end{align*}
  \begin{center}\vspace*{-2em}
  \end{center}
  \caption{Calculus rules for arbitrary traces}
  \label{fig:Arb}
\end{figure}

\paragraph{Rules for Arbitrary Traces (\Cref{fig:Arb}).}
According to Figure~\ref{fig:semantics-formulas}, chop sequences
$true \chop \Psi$ indicate an arbitrary finite trace, represented by
$true$, eventually ending with a desired result $\Psi$.  This closely
resembles the \emph{eventually} operator of LTL.
Rule \texttt{ARB1} assumes the situation that $\Psi$ already holds in
the current state, while \texttt{ARB2} assumes $\Psi$ does not hold
yet, allowing us to skip the leading formula.

\paragraph{Additional Rules.}
\iftechreport%
  Rules deemed not necessary to understand the central
  concept behind the calculus can be found in
  Appendix~\ref{sec:additional-material}.%
\else%
  Supplementary rules deemed not necessary to understand the central
  concept behind the calculus can be found in~\cite{Heidler24TR}.
\fi

\subsection{Fixed Point Induction}
\label{sec:fixed point-induction}

When encountering a fixed point operation $\mu X. \Phi$ in the
antecedent, one possible derivation strategy is repeated usage of rule
\texttt{UNFL} until the recursion terminates based on the current
program state. However, not only does a high recursion bound blow up
the proof tree size, recursion with an unknown bound may not terminate
at all. This may cause the derivation strategy to be unusable,
motivating an alternative approach.

\begin{example}
  Trace formula $Sb_x^{10} \chop \Phi_{fac}$ can be handled by a
  derivation strategy with repeated unfolding. However, this does not
  work for just $\Phi_{fac}$, because $x$ then has an unknown value,
  causing the recursion to have an unknown bound.
\end{example}

In the remaining paper we assume a convention giving a unique name to
each recursion variable. 

\begin{theorem}[Fixed Point Induction]
  \label{fpindproof}
  For recursion variables $X_1,\,X_2$, a predicate $I$, a valuation
  $\mathbb{V}$, and trace formulas
  $\mu X_{1.} \Phi,\,\mu X_{2.}  \Psi$:
  $$
  \text{If } \llbracket I \land X_1\rrbracket_{\mathbb{V}} \subseteq
  \llbracket X_2\rrbracket_{\mathbb{V}} \text{ implies } \llbracket I
  \land \Phi\rrbracket_{\mathbb{V}} \subseteq
  \llbracket\Psi\rrbracket_{\mathbb{V}}
  \text{ then }
  \llbracket I \land \mu X_{1.} \Phi\rrbracket_{\mathbb{V}} \subseteq
  \llbracket\mu X_{2.} \Psi\rrbracket_{\mathbb{V}}\enspace.
  $$
\end{theorem}

\begin{proof}
  Let recursion variables $X_1,\,X_2$, predicate $I$, valuation
  $\mathbb{V}$ and trace formulas $\mu X_{1.} \Phi$, $\mu X_{2.} \Psi$
  be arbitrary, but fixed. Since $\llbracket I \land X_1\rrbracket_\mathbb{V} = \llbracket I\rrbracket_{\mathbb{V}} \cap \mathbb{V}(X_1)$:
  \begin{align*} 
    &\llbracket I \land X_1\rrbracket_\mathbb{V} \subseteq \llbracket X_2\rrbracket_\mathbb{V} \text{ implies } \llbracket I \land \Phi\rrbracket_\mathbb{V} \subseteq \llbracket\Psi\rrbracket_\mathbb{V} 
    \\
    \Longleftrightarrow & \hspace{1mm} \forall \gamma_1, \gamma_2. \hspace{1mm} \llbracket I\rrbracket_{\mathbb{V}} \cap \gamma_1 \subseteq \gamma_2 \text{ implies } \llbracket I\rrbracket_{\mathbb{V}} \cap \llbracket\Phi\rrbracket_{\mathbb{V}[X_1 \mapsto \gamma_1]} \subseteq \llbracket\Psi\rrbracket_{\mathbb{V}[X_2 \mapsto \gamma_2]}
  \end{align*}

  We define the following $\gamma$-sequences:
  $$
  (\gamma_1^i, \gamma_2^i)_{i \geq 0} \text{ with } (\gamma_1^0, \gamma_2^0) = (\varnothing, \varnothing),\, \gamma_1^{i+1} = \llbracket \Phi\rrbracket_{\mathbb{V}[X_1 \mapsto \gamma_1^i]},\, \gamma_2^{i+1} = \llbracket\Psi\rrbracket_{\mathbb{V}[X_2 \mapsto \gamma_2^i]}
  $$

  We prove by natural induction over $i$ that
  $\llbracket I\rrbracket_{\mathbb{V}} \cap \gamma_1^i \subseteq
  \gamma_2^i$ for every $i \geq 0$. In the case $i = 0$ we have
  $\llbracket I\rrbracket_{\mathbb{V}} \cap \gamma_1^0= \llbracket
  I\rrbracket_{\mathbb{V}} \cap \varnothing = \varnothing \subseteq
  \gamma_2^0$.
  
  Assume as the induction hypothesis that
  $\llbracket I\rrbracket_{\mathbb{V}} \cap \gamma_1^i \subseteq
  \gamma_2^i$ for a fixed $i \geq 0$. Using our premise, this implies
  $\llbracket I\rrbracket_{\mathbb{V}} \cap \llbracket
  \Phi\rrbracket_{\mathbb{V}[X_1 \mapsto \gamma_1^i]} \subseteq
  \llbracket\Psi\rrbracket_{\mathbb{V}[X_2 \mapsto \gamma_2^i]}$. Then
  also
  $$
  \llbracket I\rrbracket_{\mathbb{V}} \cap \gamma_1^{i+1} = \llbracket
  I\rrbracket_{\mathbb{V}} \cap \llbracket
  \Phi\rrbracket_{\mathbb{V}[X_1 \mapsto \gamma_1^i]} \subseteq
  \llbracket\Psi\rrbracket_{\mathbb{V}[X_2 \mapsto \gamma_2^i]} =
  \gamma_2^{i+1}.
  $$

  Both sequences must ---after possibly infinitely many steps--- reach
  their least fixed points. This means that
  $\llbracket I\rrbracket_\mathbb{V} \cap \llbracket\mu X_{1.}
  \Phi\rrbracket_\mathbb{V} \subseteq \llbracket \mu X_{2.}
  \Psi\rrbracket_\mathbb{V}$ must hold. This is equivalent to our
  proof obligation
  $\llbracket I \land \mu X_{1.} \Phi\rrbracket_\mathbb{V} \subseteq
  \llbracket \mu X_{2.} \Psi\rrbracket_\mathbb{V}$.
\end{proof}

\begin{figure}[t]
\begin{align*} 
\begin{prooftree}
  \hypo{P_\Gamma \vdash I}
  \hypo{\xi, (X_1|_{I}, X_2) \s I,\, \Phi \vdash \Psi}
  \infer[left label=\texttt{FPI}]2{\xi \s \Gamma,\, \mu X_{1.} \Phi \vdash \mu X_{2.} \Psi,\, \Delta}
\end{prooftree}
\end{align*}
\begin{center}\vspace*{-2em}
\end{center}
\caption{Fixed point induction rule}
\label{fig:FPI}
\end{figure}

\paragraph{Fixed Point Induction Rule (\Cref{fig:FPI}).}

Rule \texttt{FPI} makes use of the theorem above to infer trace
inclusion between fixed point formulas. Invariant $I$ allows us to
preserve program state information for the derivation of an arbitrary
recursive iteration.
The first premise establishes that the invariant holds initially. The
second premise then takes the shape of the fixed point induction
assumption as in Theorem~\ref{fpindproof}, representing an arbitrary
recursive iteration. Note that this premise also enforces the
invariant to be preserved, as the derivation between recursion
variables $X_1$, $X_2$ can only be proven if the invariant holds in
the program state before $X_1$ (see rule \texttt{RVAR}).
\iftechreport
  \noindent An alternative fixed point rule can be found in
  Appendix~\ref{sec:additional-material}.  \else An alternative fixed
  point rule can be found in~\cite{Heidler24TR}.\footnote{This fixed
    point rule is based on an extension of the definition of triples
    $(X, p, X') \in \xi$ with more general triples
    $(X, p, \Phi) \in \xi$, where $\Phi$ matches a trace formula.}
\fi

\begin{example}\label{ex:fac-deriv}
  A derivation using rule \texttt{FPI} is in \Cref{fig:dem2}: We prove
  that the factorial program $S_{fac}$ never decreases variable $y$
  after its initialization, or else $x$ is initialized with a negative
  value. For better readability, we use abbreviations:
  $$
  \Phi_{fac}' \equiv ((x = 1 \land Id \chop Id) \lor (x \neq 1 \land Id
  \chop Sb_y^{y*x} \chop Sb_x^{x-1} \chop Id \chop X_{fac}))
  $$
  $$
  \Phi_{inc}' \equiv R_{inc}^y \lor R_{inc}^y \chop X_{inc}
  $$

  Before usage of \texttt{FPI}, trace lengthening is needed to
  synchronize trace lengths and positions of recursion variable
  occurrences. Lengthening $\Phi_{inc}'$ by a factor of three yields
  $R_{inc}^y \chop R_{inc}^y \chop R_{inc}^y \chop R_{inc}^y \chop
  X_{inc}$ as its chop sequence, which synchronizes with the right
  disjunct in $\Phi_{fac}'$. The left disjunct also synchronizes due
  to the occurrence of $R_{inc}^y \chop R_{inc}^y$.
\end{example}

\begin{figure}[t]
  \begin{center}
    \begin{prooftree}
      \infer[left label=\texttt{CLOSE}]0{x \geq 0, y = 1 \vdash \bigwedge P_\Gamma^1}
          \hypo{}
    \ellipsis{}{P_\Gamma^4 \vdash \bigwedge P_\Gamma^1}
    \infer[left label=\texttt{RVAR}]1{(X_{fac}|_{\bigwedge P_\Gamma^1}, X_{inc}) \s P_\Gamma^4, X_{fac} \vdash X_{inc}}
      \ellipsis{}{(X_{fac}|_{\bigwedge P_\Gamma^1}, X_{inc}) \s P_\Gamma^1, \Phi_{fac}' \vdash repeat_3(\Phi_{inc}')}
      \infer[left label=\texttt{FPI}]2{x \geq 0, y = 1, \mu X_{fac.} \Phi_{fac}' \vdash \mu X_{inc.} repeat_3(\Phi_{inc}')}
      \infer[left label=\texttt{LENR}]1[$3 \geq 1$]{x \geq 0, y = 1, \mu X_{fac.} \Phi_{fac}' \vdash \mu X_{inc.} \Phi_{inc}'}
      \ellipsis{}{x \geq 0, Sb_y^1 \chop Id \chop \mu X_{fac.} \Phi_{fac}' \vdash Sb_y^1 \chop \mu X_{inc.} \Phi_{inc}'}
      \ellipsis{}{Sb_y^1 \chop Id \chop \mu X_{fac.} \Phi_{fac}' \vdash Sb_y^1 \chop \mu X_{inc.} \Phi_{inc}' \lor x < 0}
    \end{prooftree}
    \begin{center}\vspace*{-2em}
    \end{center}
    \caption{Demonstration of fixed point induction}
    \label{fig:dem2}
  \end{center}
\end{figure}

\begin{theorem}[Soundness]
  \label{th:s1}
  The calculus rules presented in this section are sound, implying
  that only valid sequents are derivable.
\end{theorem}

\iftechreport
\noindent Due to its length, the soundness proof has been moved to Appendix~\ref{sec:additional-material}.
\else
\noindent The proof of this theorem is in \cite{Heidler24TR} due to its length.
\fi

%%% Local Variables:
%%% mode: latex
%%% TeX-master: "main"
%%% End:

\section{Calculus Extensions}
\label{ch5}

% In this section, we work out problems in our calculus and present an
% approach for their solution via calculus extensions.

\subsection{Contracts}
\label{sec:contracts}

The base rules of the calculus we established so far expose a major
source of incompleteness: If in an antecedent the fixed point
operation or the recursion variable occurs \emph{non-tail
  recursively}, such as in $X \chop \Phi$ or
$(\mu X. \Psi) \chop \Phi$, then there is no rule to continue a
derivation. The root cause is that the effect that a fixed point or a
recursion variable has on the execution state is unknown. For this
reason, all the rules dealing with fixed points so far permit only a
single formula in the antecedent.
% This stems from the fact that the program state modification during
% such a step is unknown, although crucial for restricting its
% antecedent traces. We will integrate procedure contracts into our
% calculus to remedy this issue.
The standard solution in deductive verification to deal with such a
situation are \emph{contracts} \cite{Hähnle2019} that summarize the
execution state after a complex statement.

% \begin{definition}[Restricted Trace Inclusion]
%   For any state $s$, let $s|_v$ be $s$ projected onto program
%   variable set $v \subseteq PVar$. Similarly, this projection can be
%   lifted to traces~$\sigma$ and sets of traces $tr$. We introduce a
%   \textit{restricted trace inclusion relation}~$\subseteq_{v}$, s.t.
% $$ tr \subseteq_{v} tr' \text{ iff } tr|_{v} \subseteq tr'|_{v}$$
% We will lift this restricted trace inclusion from $\subseteq_{v}$ to
% $\models_v$.
%\end{definition}

%\noindent Note that the trace inclusion $tr \subseteq tr'$ trivially
% implies the restricted trace inclusion $tr \subseteq_{v} tr'$ for
% any $v \subseteq PVar$. Furthermore, it is interesting to realize
% that any trace consequence $\Phi \models \Psi$ holds iff
% $\Phi \models_{v} \Psi$ for
% $pvar(\Phi) \cup pvar(\Psi) \subseteq v$.

\begin{definition}[Procedure Contract]
  A state-based \emph{procedure contract} for a given trace formula
  $\Phi$ is a pair $(pre, post)$ of \textit{precondition}
  $pre \in Pred$ and \textit{postcondition} $post \in
  Pred$. Postconditions may contain fresh program variables $x_{old}$
  containing the value of variables $x$ in $\Phi$ in the execution
  state before $\Phi$ is evaluated.
\end{definition}

While contracts may approximate \emph{any} kind of trace formula, we
kept the attribute ``procedure'', because the trace formula of a
contract can be thought of as the body of a procedure declaration and
this is also how we use contracts.
Intuitively, a procedure contract $(pre, post)$ is \emph{valid} for a
trace formula $\Phi$, if the postcondition is satisfied in the
execution state after evaluation of $\Phi$, assuming the precondition
is satisfied in the execution state before evaluation of $\Phi$.

\begin{example}\label{ex:contract}
  A valid procedure contract for trace formula $\Phi_{m}$ in
  \Cref{ex:fac} is
  $$(x \geq 1,\, y = y_{old} * x_{old}! \land x = 1)\enspace.$$
\end{example}

We \emph{encode} the intuitive validity of a procedure contract
formally as trace inclusion.

\begin{definition}[Contract Encoding] Let $(v^i)_{1 \leq i \leq n}$ be
  all program variables occurring in $\Phi$ and
  $(v_{old}^i)_{1 \leq i \leq n}$ \emph{fresh} program variables. A procedure contract
  $(pre, post)$ is \emph{valid} for $\Phi$ in $\mathbb{V}$ iff
  $$
  \llbracket\underbrace{\bigwedge v_{old}^i = v^i \land pre \land \Phi \chop true}_{\langle pre(\Phi)\rangle}\rrbracket_{\mathbb{V}} \subseteq \llbracket \underbrace{\Phi \chop post}_{\langle post(\Phi)\rangle}\rrbracket_{\mathbb{V}}\enspace.
  $$
\end{definition}

In the following, we use abbreviations $\langle pre(\Phi)\rangle$,
$\langle post(\Phi)\rangle$ for the encoding of the pre- and
postcondition, respectively, as indicated above.
The encoding expresses: Assuming precondition $pre$ holds and the
information about the execution state \emph{before the evaluation}
of $\Phi$ is memorized using fresh variables $v_{old}^i$, then after
evaluating $\Phi$ we reach a state in the antecedent that implies
$post$ in the succedent.  Observe that to model this as a trace
inclusion formula, we have to copy the formula $\Phi$ into the
succedent to ensure that the traces match.
% until the evaluation of $\Phi$ is completed.

\begin{theorem}[Fixed Point Induction on Contracts] \label{th:1} For
  any recursion variable $X$, trace formula $\Phi$, valuation
  $\mathbb{V}$, and procedure contract $(pre, post)$, if the validity
  of $(pre, post)$ for $X$ in $\mathbb{V}$ implies its validity for
  $\Phi$ in $\mathbb{V}$, then it must also be valid for $\mu X. \Phi$
  in $\mathbb{V}$.
\end{theorem}

\iftechreport
\noindent \noindent The proof for this theorem is in Appendix~\ref{app:contract}.
\else
\noindent The proof for this theorem is in~\cite{Heidler24TR}.
\fi

To integrate contracts into the calculus rules presented in
Section~\ref{ch4}, we need to remodel sequents so they include
information about procedure contracts.

\begin{definition}[Sequent with Contract]
  A \emph{procedure contract table} is a partial function
  $\mathbb{C} : ProcName \rightharpoonup Pred \times Pred$, assigning
  each procedure of a program $P$ a possible contract. $\mathbb{C}$ is called
  \emph{valid} in $\mathbb{V}$ iff for all $m \in dom(\mathbb{C})$,
  $\mathbb{C}(m)$ is valid for $\mu X_{m.} \Phi$ in~$\mathbb{V}$, where $\mu X_{m.} \Phi$ is the subformula of $\Gamma$ corresponding to procedure $m$.
  A \emph{sequent (with contract)} has the form
  $\xi \s \Gamma \vdash_{\mathbb{C}} \Delta$, where a procedure
  contract table $\mathbb{C}$ is added as an index to $\vdash$.
\end{definition}

Note that procedure contracts in our sequents are only available for fixed point formulas $\mu X_{m.} \Phi$ generated by procedures $m$ via $\mathit{stf(P)}$, which is sufficient for proving sequents of the form $\mathit{stf}(P) \vdash_{\mathbb{C}} \Psi$. % An extension of method contracts to encompass arbitrary fixpoint operations is considered future work.

\begin{definition}[Validity of Sequent with Contract]
  A sequent $\xi \s \Gamma \vdash_{\mathbb{C}} \Delta$ is
  \emph{valid}, if for all valuations $\mathbb{V}$, contract table
  $\mathbb{C}$ valid in $\mathbb{V}$, and
  $\llbracket X \land p\rrbracket_{\mathbb{V}} \subseteq \llbracket
  X'\rrbracket_{\mathbb{V}}$ holding for all $(X|p, X') \in \xi$
  implies
  $\llbracket \bigwedge\Gamma\rrbracket_{\mathbb{V}} \subseteq
  \llbracket\bigvee\Delta\rrbracket_{\mathbb{V}}$.
\end{definition}

The contract table $\mathbb{C}$ is always empty in a top-level sequent
of a derivation.  Procedure contracts are added to $\mathbb{C}$ on
demand by the calculus rules during a derivation. The rules ensure
that all added contracts are proven valid.

\begin{example}
  Continuing Example~\ref{ex:contract}, let
  $\mathbb{C}(fac) \equiv (x \geq 1, y = y_{old} * x_{old}! \land x =
  1)$.
  $$
  P_\Gamma^2,\, \Phi_{m} \chop Sb_x^{x-1} \vdash_{\mathbb{C}} true \chop
  x = 0
  $$
  is a valid sequent, because the postcondition guarantees that $fac$
  terminates with $x = 1$ before eventually being reduced to $x = 0$.
\end{example}

\begin{figure}[t]
  \begin{align*} 
    & \begin{prooftree}
      \hypo{v_{old}^i \in fresh(Var)}
      \hypo{\mathbb{C}' = \mathbb{C}[m \mapsto (pre, post)]}
      \hypo{\xi \s \langle pre(\Phi)\rangle \vdash_{\mathbb{C}'} \langle post(\Phi)\rangle}
      \hypo{\xi \s \Gamma, \mu X_{m.} \Phi \vdash_{\mathbb{C}'} \Delta}
      \infer[left label=\texttt{MC}]{2, 2}{\hspace{22mm}\xi \s \Gamma, \mu X_{m.} \Phi \vdash_{\mathbb{C}} \Delta \hspace{22mm}}
    \end{prooftree} \\[7pt]
    & \begin{prooftree}
      \hypo{v_{old}^i \in fresh(Var)}
      \hypo{\mathbb{C}' = \mathbb{C}[m \mapsto (pre, post)]}
      \hypo{\xi \s \langle pre(\Phi_1)\rangle \vdash_{\mathbb{C}'} \langle post(\Phi_1)\rangle}
      \hypo{\xi \s \Gamma, (\mu X_{m.} \Phi_1) \chop \Phi_2 \vdash_{\mathbb{C}'} \Delta}
      \infer[left label=\texttt{CH-MC}]{2, 2}{\hspace{23mm} \xi \s \Gamma, (\mu X_{m.} \Phi_1) \chop \Phi_2 \vdash_{\mathbb{C}} \Delta \hspace{23mm}}
    \end{prooftree}
  \end{align*}
  \begin{center}
    \vspace*{-2em}
  \end{center}
  \caption{Calculus rules for procedure contract validity}
  \label{fig:mce}
\end{figure}

\paragraph{Procedure Contract Validity Rules (\Cref{fig:mce}).}

Rules \texttt{MC} and \texttt{CH-MC} prove the validity of a procedure
contract for the leading fixed point formula and add it to the
procedure contract table $\mathbb{C}$, as can be seen in the right
premise.
The left premise assumes the procedure contract holds for the
internal recursion variable $X_m$ and proves that it hence must also
be valid for $\Phi$, $\Phi_1$. \Cref{th:1} justifies the validity of
the contract for the whole fixed point formula $\mu X_{m.}  \Phi$.
The proof uses contract table $\mathbb{C}'$ that already assumes the
contract for $m$, because this contract may be assumed to handle
recursive calls to $m$ in $\Phi$, $\Phi_1$.

\begin{figure}[t]
  \begin{align*} 
    & \begin{prooftree}
      \hypo{\mathbb{C}(m) = (pre, post)}
      \hypo{P_\Gamma \vdash_{\mathbb{C}} p \land pre}
      \hypo{\xi \s P_\Gamma[v_{old}^i/v^i],\, post,\, \Phi \vdash_{\mathbb{C}} \Psi}
      \infer[left label=\texttt{CH-RVAR}]{1, 2}{\hspace{6mm} \xi,\, (X_{m}|_{p}, X) \s \Gamma,\, X_m \chop \Phi \vdash_{\mathbb{C}} X \chop \Psi,\, \Delta \hspace{6mm}}
    \end{prooftree} \\[7pt]
    & \begin{prooftree}
      \hypo{\mathbb{C}(m) = (pre, post)}
      \hypo{P_\Gamma \vdash_{\mathbb{C}} I \land pre}
      \hypo{\xi,\, (X_m|_{I}, X) \s I,\, \Phi_1 \vdash_{\mathbb{C}} \Psi_1}
      \hypo{\xi \s P_\Gamma[v_{old}^i/v^i],\, post,\, \Phi_2 \vdash_\mathbb{C} \Psi_2}
      \infer[left label=\texttt{CH-FPI}]{1, 3}{\hspace{26mm} \xi \s \Gamma, (\mu X_{m.} \Phi_1) \chop {\Phi_2} \vdash_{\mathbb{C}} (\mu X. \Psi_1) \chop \Psi_2, \Delta \hspace{26mm}}
    \end{prooftree}
  \end{align*}
  \begin{center}
    \vspace*{-2em}
  \end{center}
  \caption{Calculus rules for procedure contract application}
  \label{fig:mca}
\end{figure}

\paragraph{Procedure Contract Application Rules (\Cref{fig:mca}).}

Rule \texttt{CH-RVAR} handles the occurrence of a recursion variable
$X_m$ in a non-tail recursive setting. In addition to rule
\texttt{RVAR}, it looks up the procedure contract $(pre, post)$ of
$m$, as indicated by the side condition. Since the recursion variable of procedure $m$ is uniquely named as $X_m$, the correct procedure is used.
The left premise additionally proves the precondition $pre$. The
right premise takes the current program state, substitutes every
occurrence of variable $v^i$ with variable $v_{old}^i$, as determined
in the contract, and adds the postcondition $post$. This modified
program state is then used to continue the derivation of the remaining
trace. Rule \texttt{CH-FPI} behaves similarly, guaranteeing the
derivation of non-tail recursive fixed point formula occurrences.

It is future work to extend the calculus to support multiple contracts
for procedures by applying contracts in a hierarchical
fashion. This necessitates a modification of the contract table
definition and the calculus rules.

% Additional contract application rules can be found in Appendix B.1.

\begin{figure}[t]
\begin{center}
\begin{prooftree}
\hypo{}
\ellipsis{}{}
\hypo{}
\ellipsis{}{}
\infer[left label=\texttt{CLOSE}]0{y \geq 1, x = 0, y > x \vdash_{\mathbb{C}} y > x}
\infer[left label=\texttt{PRED}]1{y \geq 1, x = 0 \vdash_{\mathbb{C}} y > x}
\ellipsis{}{\hspace{-3mm} x_{old} \geq 1, y_{old} \geq 1, y = y_{old} * x_{old}! \land x = 1, Sb_x^{x-1} \vdash_{\mathbb{C}} Sb_x^{x-1} \chop y > x}
\infer[left label=\texttt{\hspace{-2mm}CH-FPI}]2{P_\Gamma^1, \mu X_{fac.} \Phi_{fac}' \chop Sb_x^{x-1} \vdash_{\mathbb{C}} \mu X_{inc.} \Phi_{inc}' \chop Sb_x^{x-1} \chop y > x}
\infer[left label=\texttt{CH-MC}]2{P_\Gamma^1, \mu X_{fac.} \Phi_{fac}' \chop Sb_x^{x-1} \vdash_{\varnothing} \mu X_{inc.} \Phi_{inc}' \chop Sb_x^{x-1} \chop y > x}
\end{prooftree}
\begin{center}\vspace*{-2em}
\end{center}
\caption{Demonstration of calculus with procedure contracts}
\label{fig:dem3}
\end{center}
\end{figure}

\begin{example}
  The calculus with procedure contracts is illustrated by an example
  in \Cref{fig:dem3}. We use the abbreviations from
  Example~\ref{ex:fac-deriv}, $\mathbb{C} := [fac \mapsto (pre, post)]$, $P_{\Gamma}^1 \equiv \{x \geq 1, y \geq 1\}$ and
  $(pre, post) \equiv (x \geq 1, y = y_{old} * x_{old}! \land x = 1)$. For readability, the derivation only follows the rightmost premises.

\end{example}

\begin{theorem}[Soundness of the Calculus with Procedure
  Contracts] \label{th4} The calculus rules presented in this section
  are sound, implying that only valid sequents are derivable.
\end{theorem}

\iftechreport
\noindent \noindent Due to its length, the soundness proof has been moved to Appendix~\ref{app:contract}.
\else
\noindent The proof of this theorem is in \cite{Heidler24TR} due to its length.
\fi

\subsection{Synchronization}
\label{sec:synchronization}

To successfully perform a fixed point induction, the trace lengths and
positions of the recursion variable occurrences must align in
antecedent and succedent. This is not always the case, and it
motivates the following synchronization rules.

\begin{figure}[t]
  \begin{center}
    \begin{prooftree}
      \hypo{\textit{not derivable}}
      \ellipsis{}{P_\Gamma^1, Sb_y^{y*x} \vdash X_{inc}}
      \hypo{\textit{not derivable}}
      \ellipsis{}{(X_{fac}|_{\bigwedge P_\Gamma^1}, X_{inc}) \s P_\Gamma^4, X_{fac} \vdash R_{inc}^y}
      \ellipsis{}{(X_{fac}|_{\bigwedge P_\Gamma^1}, X_{inc}) \s P_\Gamma^3, Sb_x^{x-1} \chop X_{fac} \vdash R_{inc}^y \chop R_{inc}^y}
      \infer[left label = \texttt{CH-UPD}]2{(X_{fac}|_{\bigwedge P_\Gamma^1}, X_{inc}) \s P_\Gamma^2, Sb_y^{y*x} \chop Sb_x^{x-1} \chop X_{fac} \vdash X_{inc} \chop R_{inc}^y \chop R_{inc}^y}
    \end{prooftree}
    \begin{center}\vspace*{-2em}
    \end{center}
    \caption{Demonstration of recursion variable synchronization problem}
    \label{fig:dem4}
  \end{center}
\end{figure}

\begin{example}
  In fixed point formula
  $\Phi_{inc}:=\mu X_{inc.} (R_{inc}^y \lor X_{inc} \chop R_{inc}^y)$,
  the recursion variable $X_{inc}$ does not occur tail recursively. So
  any synchronizing formula must have its recursion variable as a
  leading formula in its chop sequence. This issue is demonstrated in
  \Cref{fig:dem4}: The second disjunct in $\Phi_{inc}$ is expanded to
  $X_{inc} \chop R_{inc}^y \chop R_{inc}^y$, so that in the initial sequent
  of \Cref{fig:dem4} the positions of recursion variables $X_{fac}$, $X_{inc}$
  misalign.
\end{example}

\begin{definition}[Chop Formula]
  Let relation $R$ and recursion variable $X$ be
  fixed. \emph{Primitive chop formulas} are a subclass of trace
  formulas consisting of chop sequences containing exclusively $R$ or
  $X$, specified by the grammar
  $$
  \Psi_{(R, X)} ::= R {\normalfont\text{ | }} X {\normalfont\text{ |
    }} \Psi_{(R, X)} \chop \Psi_{(R, X)}\enspace.
  $$
  The \emph{chop formulas} $CF_{(R, X)}$ with fixed $R$ and $X$ are
  defined as disjunctions over primitive chop formulas, specified by
  the grammar
  $$
  \Phi_{(R, X)} ::= \Psi_{(R, X)} {\normalfont\text{ | }} \Psi_{(R, X)} \lor \Phi_{(R, X)}\enspace.
  $$
  All recursion variables $X$ occurring in a chop formula are
  \emph{not bound}.
\end{definition}

\begin{example}\label{ex:gr}
  %The formula
  $
  \Phi_{sub} \equiv \mathit{Id} \lor \mathit{Id} \chop X \chop \mathit{Id} \chop X \lor \mathit{Id} \chop \mathit{Id} \chop \mathit{Id}
  $
  is a chop formula, i.e.\ $\Phi_{sub} \in CF_{(\mathit{Id}, X)}$. The
  subformula $\mathit{Id} \chop X \chop \mathit{Id} \chop X$ is a primitive chop
  formula.
\end{example}

Let $\Phi \in CF_{(R, X)}$ be a chop formula. Then there exists a
natural mapping
$ gr{\,:\,}CF_{(R, X)} \rightarrow G(\{X\}, \{R\}, \delta, X) $ from
$\Phi = \bigvee_{1 \leq i \leq n} \varphi_i$ to a context-free grammar
with \emph{non-terminal} $X$, \emph{terminal} $R$, \emph{production
  rules} $\delta$ and \emph{initial non-terminal} $X$, where
production rules $\delta$ are defined as
$X \rightarrow grammatize(\varphi_i)$ for $ 1 \leq i \leq n$.
The function $grammatize$ maps each primitive chop formula to a
sequence over terminal $R$ and non-terminal $X$. It is defined by
$$
grammatize(S_1 \chop S_2 \chop \ldots \chop S_n) := S_1 S_2 \cdots S_n
\text{ for } S_i\in\{R,\,X\}\enspace.
$$

This construction ensures that every $\Phi \in CF_{(R, X)}$ has a
\emph{unique grammar representation} $gr(\Phi)$.  There is exactly one
terminal symbol in $gr(\Phi)$, so we may use Parikh's theorem
\cite{parikh} to deduce that its specified language is regular.
% There also exist approaches to transform this kind of context-free
% grammar into an NFA accepting the same language~\cite{Esparza_2011}.

\begin{definition}%[Trace Language]
  The \emph{regular trace language} of a chop formula $\Phi$ is
  $L(gr(\Phi))$.
\end{definition}

\begin{example}
  The context-free grammar $gr(\Phi_{sub})$ of the formula from
  Example~\ref{ex:gr} is:
  $ X \rightarrow \mathit{Id} \text{ | } \mathit{Id} \hspace{0.5mm} X \hspace{0.5mm} \mathit{Id}
  \hspace{0.5mm} X \text{ | } \mathit{Id} \hspace{0.5mm} \mathit{Id} \hspace{0.5mm} \mathit{Id}
  $.  Now consider the chop formula
  $\Phi_{sub}' \equiv \mathit{Id} \lor \mathit{Id} \chop \mathit{Id} \chop X \chop X \lor \mathit{Id}
  \chop Id \chop Id$.  Its context-free grammar $gr(\Phi_{sub}')$ has
  the production rules:
  $
  X \rightarrow \mathit{Id} \text{ | } \mathit{Id} \hspace{0.5mm} \mathit{Id} \hspace{0.5mm} X \hspace{0.5mm} X \text{ | } \mathit{Id} \hspace{0.5mm} \mathit{Id} \hspace{0.5mm} \mathit{Id}
  $.
  The induced regular trace languages are identical, i.e.\
  $L(\Phi_{sub}) = L(\Phi_{sub}')$, implying that both chop formulas
  generate the exact same traces.
\end{example}

\begin{figure}[t]
  \begin{align*} 
    & \begin{prooftree}
      \hypo{\xi \s \Gamma \vdash \mu X. \Psi',\, \Delta}
      \infer[left label=\texttt{SYNC}]1[$L(gr(\Psi')) \subseteq L(gr(\Psi))$]{\xi \s \Gamma \vdash \mu X. \Psi,\, \Delta}
    \end{prooftree} 
  \end{align*}
  \begin{center}\vspace*{-2em}
  \end{center}
  \caption{Calculus rule for $\mu$-formula synchronization}
  \label{fig:sync}
\end{figure}

\paragraph{Synchronization Rule (\Cref{fig:sync}).}

Rule \texttt{SYNC} permits to realign problematic fixed point formulas
to synchronize with the antecedent. This requires the trace language
of the premise to be smaller than or equal to the trace language of
the conclusion.  We cannot apply the synchronization rule when the
fixed point formula in the premise is not a chop formula (for example,
in the case of nested fixed point formulas), which is a limitation to
completeness.

\begin{figure}[t]
  \begin{center}
    \begin{prooftree}
      \hypo{\textit{See \Cref{fig:dem1} (cf.\ \Cref{fig:dem4})}}
      \ellipsis{}{(X_{fac}|_{\bigwedge P_\Gamma^1}, X_{inc}) \s P_\Gamma^2, Sb_y^{y*x} \chop Sb_x^{x-1} \chop X_{fac} \vdash R_{inc}^y \chop R_{inc}^y \chop X_{inc}}
      \ellipsis{}{\Phi_m \vdash \mu X_{inc.} (R_{inc}^y \lor R_{inc}^y \chop X_{inc})}
      \infer[left label = \texttt{SYNC}]1{\Phi_m \vdash \mu X_{inc.} (R_{inc}^y \lor X_{inc} \chop R_{inc}^y)}
    \end{prooftree}
    \begin{center}\vspace*{-2em}
    \end{center}
    \caption{Demonstration of $\mu$-formula synchronization}
    \label{fig:dem5}
  \end{center}
\end{figure}

% An additional synchronization rule can be found in Appendix C.1.

\begin{example}
  A derivation with $\mu$-formula synchronization is in
  \Cref{fig:dem5}.
\end{example}

\begin{theorem}[Soundness of the Calculus with Synchronization]
  \label{t2}
  The \emph{\texttt{SYNC}} rule is sound, implying that only valid sequents
  are derivable.
\end{theorem}

\iftechreport
\noindent \noindent Due to its length, the soundness proof has been moved to Appendix~\ref{app:sync}.
\else
\noindent The proof of this theorem is in \cite{Heidler24TR} due to its length.
\fi

%%% Local Variables:
%%% mode: latex
%%% TeX-master: "main"
%%% End:

\section{Related Work}
\label{ch6}

Lange et al. \cite{RW1} analyze the model checking problem over finite
transition systems using a modal $\mu$-calculus logic enriched with a
chop operator. They focus on providing a model checker for this
extended logic and prove its soundness and completeness.  The paper
presents a tableau calculus that lets one verify whether a transition
system $T$ satisfies a corresponding formula $\Phi$. Formula
consequence is \emph{not} addressed.

Walukiewicz \cite{RW2} extends propositional modal logic with fixpoint
operations, resulting in the common $\mu$-calculus. An axiomatization
is provided to syntactically infer sequents $\Gamma \vdash \Delta$
that semantically correspond to the implication between $\mu$-calculus
formulas. The presented calculus is proven to be \emph{sound} and
\emph{complete}.  In contrast to the present paper, the logic syntax
contains modal connectives, but neither relations nor the chop
operator.

Müller-Olm \cite{RW3} extends the classical modal $\mu$-calculus with
chop, which is semantically interpreted using \emph{predicate
  transformers}. The paper focuses on proving that any context-free
process has a characteristic formula up to bisimulation or simulation.
The paper further analyzes decidability and expressiveness of this
logic, but reasoning about formula consequence is \emph{not}
discussed.

%%% Local Variables:
%%% mode: latex
%%% TeX-master: "main"
%%% End:

\section{Conclusion}
\label{ch7}

We designed a sound calculus to prove formula consequence in a trace
logic with smallest fixed points, chop, and binary relations. The
significance of the logic derives from the fact that it can
characterize the behavior of imperative programs with recursive
procedures. To prove the judgment $S{\,:\,}\Phi$ that a program~$S$
conforms to a trace formula specification $\Phi$, it is necessary to
infer consequence relations $\Phi \models \Psi$ of trace
formulas~\cite{GurovHaehnle24}.

As usual for a logic with smallest fixed point operator, the calculus
presented here has fixed point induction as its central inference
rule, but in its standard form this turns out not to be very
useful. The reason is the presence of the chop operator which
(i)~necessitates to approximate the state \emph{after} evaluation of
the first constituent in a chop formula and (ii)~may cause
misalignment among the bodies of smallest fixed point formulas. We
added \emph{contracts} for fixed point formulas and grammar-based
realignment, respectively, to mitigate these issues. We have not seen
such mechanisms in the literature on proof systems related to
$\mu$-calculus and believe these ideas constitute an interesting and
viable approach to make such calculi more complete.

At the same time, both presented solutions are clearly incomplete:
Regarding~(i), consequence between fixed points with unbounded
iterations and a formula like $true \chop \Phi$ cannot be proven: This
requires to track state changes \emph{during} the fixed point
evaluation, between iterations. Related to~(ii), $\mu$-formula
synchronization was defined for a specific subclass of trace
formulas. Direct generalization of grammar-based alignment leads to
the inclusion problem of context-free grammars which is undecidable.

In the future we want to investigate how the novel
concepts ---contracts and grammar-based alignment--- can be generalized
towards completeness and how they can be employed in automated proof
search.
% An implementation would be helpful to explore such calculi and to
% provide more complex examples.  In order to solve these problems, we
% envision further extensions of this calculus with even more
% sophisticated concepts.
It is also interesting to analyze the practicality of an integration
of this calculus with related calculi relying on trace-based judgments
\cite{hoho,HSK23a}.

%%% Local Variables:
%%% mode: latex
%%% TeX-master: "main"
%%% End:

\bibliographystyle{splncs04}
\bibliography{main}

\iftechreport
\appendix
\section{Additional Examples}
\label{appendix}

In addition to the running example used throughout the paper, we
succeeded to prove several non-trivial, interesting properties of
programs.  The  proofs  are  executed  in  the  calculus  for  judgments $S{\,:\,}\Phi$  in  \cite{GurovHaehnle24},  while  necessary  weakening steps  were  proven  in  the  calculus  presented  here.    The  derivations  can be  found  in  \cite{Heidler24}.

\begin{enumerate}
\item Let program $S_{down}$ be a program that decreases a variable
  $x$ by $2$ until $x$ reaches the value $0$. Afterwards, it further
  decreases  variable  $x$  by  $1$.  Whether  the  recursion  is  entered depends  on  the  initial  value  of  $x$.
 	\begin{align*}
 	        S_{down}  \equiv\,  &  down()  \text{  with  }\\	                                &  down  \{\text{\textbf{if}  }  x  =  0  \text{  \textbf{then}  }  x  :=  x  -  1  \text{  \textbf{else}  }  x  :=  x  -  2;  down()\}
 	    \end{align*}
  The following properties of this program were proven:
  \begin{enumerate}
  \item Variable $x$ never increases through the program execution:
    $$
    \mu X_{dec.} R_{dec}^x \lor R_{dec}^x \chop X_{dec}
    $$
    with
    $R_{dec}^x := \{(s, s') \in State \times State \text{ | }
    \mathbb{A}\llbracket x\rrbracket(s) \geq \mathbb{A}\llbracket
    x\rrbracket(s')\}$.
  \item If $x$ is \textit{even} and \textit{non-negative}, then $x$
    will eventually reach value $0$. Afterwards, $x$ will eventually
    reach value $-1$:
    $$
    \overline{even(x)} \lor x < 0 \lor true \chop x = 0 \chop x = -1
    $$
  \end{enumerate}

\item Let Program $S_{fac}$ compute the factorial of $10$ and store
  the result of the computation in variable $y$.
  \begin{align*}
S_{fac}  \equiv  \hspace{1mm}  &  x  :=  10;  y  :=  1;  factorial()  \text{  with  }  \\
                            & factorial \{\text{\textbf{if} } x = 1 \text{ \textbf{then} } skip \text{ \textbf{else} } y := y * x; x := x - 1; factorial()\}
  \end{align*}

  The following property of this program was proven:\smallskip

  Variable $y$ will eventually map to $10!$: \quad $true \chop y = 10!$

\item Let program $S_{pow}$ compute the power $y^x$ and store the
  result in variable $z$. This  is  a  program  with  mutually  recursive  procedures.
  \begin{align*}
    S_{pow}  \equiv  \hspace{1mm}  &  z  :=  1;  pow()  \text{  with}  \\
                            & pow \{\text{\textbf{if} } x = 1 \text{ \textbf{then} } skip \text{ \textbf{else} } z := z * y; subtract()\} \text{ and } \\
                            & subtract \{ x := x - 1; pow() \}
  \end{align*}

  The following property of this program was proven:\smallskip

  Either variable $z$ never changes after its initialization or
  variable $x$ will eventually change:
  \begin{align*}
    & (Sb_z^1 \chop \mu X_{zstat.} (R_{stat}^z \lor R_{stat}^z \chop X_{zstat})) \hspace{1mm} \lor \\
    & (\mu X_{xstat.} R_{stat}^x \lor R_{stat}^x \chop X_{xstat}) \chop R_{change}^x \chop true
  \end{align*}
  $$
  \text{with } R_{stat}^x := \{(s, s') \text{ | } s(x) = s'(x)\}, R_{change}^x := \{(s, s') \text{ | } s(x) \neq s'(x)\}.
  $$
  
\item Let $S_{contract}$ be a program behaving as follows. If $x$ is
  $0$, the program terminates. If $x > 0$, then $x$ is decreased by
  $1$, before the method is called recursively and $ev$ is set to
  $0$. If $x < 0$, then $x$ is increased by $1$, before the method is
  called recursively and $ev$ is set to $1$.
  This  is  an  example  of  a  non-linear,  non-tail  recursive  program  with unbounded  behavior.	    \begin{align*}
 S_{contract}  \equiv  \hspace{1mm}  &  main()  \text{  with}  \\ 	                                                                  &  main  \{\text{\textbf{if}  }  x  =  0  \text{  \textbf{then}  }  skip  \text{  \textbf{else}}  \\	                                                                  &  \hspace{15mm}  \text{\textbf{  if}  }  x  >  0  \\  &  \hspace{20mm}  \text{  \textbf{then}  }  x  :=  x  -  1;  main();  ev  :=  0  \\	                                                                  &  \hspace{20mm}  \text{  \textbf{else}  }  x  :=  x  +  1;  main();  ev  :=  1  \}
  \end{align*}

  The following property of this program was proven:\smallskip

  At some point a state is reached where $ev$ is $0$ or $ev$ is 1 and
  $x$ is $0$ assuming $x$ is initialized with $x \neq 0$:
  $$x = 0 \lor true \chop (ev = 0 \lor ev = 1) \land x = 0$$
\end{enumerate}

%%% Local Variables:
%%% mode: latex
%%% TeX-master: "main"
%%% End:

\section{Additional Material Relating to Section~\ref{ch4}}
\label{sec:additional-material}

\subsection{Additional Base Rules}
\label{sec:additional-rules}

\begin{figure}[h]
  \begin{align*} 
    &\begin{prooftree}
      \hypo{\xi \s \Gamma, p, p \chop \Phi \vdash \Delta}
      \infer[left label=\texttt{CH-PREDL}]1{\xi \s \Gamma, p \chop \Phi \vdash \Delta}
    \end{prooftree} 
    &&\begin{prooftree}
      \hypo{\xi \s \Gamma, q \vdash true \chop \Psi, \Delta}
      \infer[left label=\texttt{CH-PREDR}]1{\xi \s \Gamma, q \vdash q \chop \Psi, \Delta}
    \end{prooftree} \\[7pt]
    &\begin{prooftree}
      \hypo{P_\Gamma, Id \vdash \Psi_1}
      \hypo{P_\Gamma \vdash \Psi_2}
      \infer[left label=\texttt{END-ID}]2{\xi \s \Gamma, Id \vdash \Psi_1 \chop \Psi_2}
    \end{prooftree}
    &&\begin{prooftree}
      \hypo{P_\Gamma, Sb_x^a \vdash \Psi_1}
      \hypo{spc_{x := a}(P_\Gamma) \vdash \Psi_2}
      \infer[left label=\texttt{END-UPD}]2{\xi \s \Gamma, Sb_x^a \vdash \Psi_1 \chop \Psi_2}
    \end{prooftree}\\[7pt]
    &\begin{prooftree}
      \hypo{\xi \s \Gamma, \Phi_1 \chop \Phi_3 \vdash \Delta}
      \hypo{\xi \s \Gamma, \Phi_2 \chop \Phi_3 \vdash \Delta}
      \infer[left label=\texttt{CH-$\lor$L}]2{\xi \s \Gamma, (\Phi_1 \lor \Phi_2) \chop \Phi_3 \vdash \Delta}
    \end{prooftree}
    &&\begin{prooftree}
      \hypo{\xi \s \Gamma, \Phi_1 \chop \Phi_3, \Phi_2 \chop \Phi_3 \vdash \Delta}
      \infer[left label=\texttt{CH-$\land$L}]1{\xi \s \Gamma, (\Phi_1 \land \Phi_2) \chop \Phi_3 \vdash \Delta}
    \end{prooftree} \\[7pt]
&
    &&\begin{prooftree}
      \hypo{\xi \s \Gamma \vdash \Psi_1 \chop \Psi_3, \Psi_2 \chop \Psi_3, \Delta}
      \infer[left label=\texttt{CH-$\lor$R}]1{\xi \s \Gamma \vdash (\Psi_1 \lor \Psi_2) \chop \Psi_3, \Delta}
    \end{prooftree} \\[7pt]
    &\begin{prooftree}
      \hypo{\xi \s \Gamma, (\Phi[\mu X. \Phi/X]) \chop \Phi' \vdash \Delta}
      \infer[left label=\texttt{CH-UNFL}]1{\xi \s \Gamma, (\mu X. \Phi) \chop \Phi' \vdash \Delta}
    \end{prooftree} 
    &&\begin{prooftree}
      \hypo{\xi \s \Gamma \vdash (\Psi[\mu X. \Psi/X]) \chop \Psi', \Delta}
      \infer[left label=\texttt{CH-UNFR}]1{\xi \s \Gamma \vdash (\mu X. \Psi) \chop \Psi', \Delta}
    \end{prooftree} \\[7pt]
    &\begin{prooftree}
      \hypo{\xi \s \Gamma, (\mu X. repeat_i(\Phi)) \chop \Phi' \vdash \Delta}
      \infer[left label=\texttt{CH-LENL}]1[$i \geq 1$]{\xi \s \Gamma, (\mu X. \Phi) \chop \Phi' \vdash \Delta}
    \end{prooftree} 
    &&\begin{prooftree}
      \hypo{\xi \s \Gamma \vdash (\mu X. repeat_i(\Psi)) \chop \Psi', \Delta}
      \infer[left label=\texttt{CH-LENR}]1[$i \geq 1$]{\xi \s \Gamma \vdash (\mu X. \Psi) \chop \Psi', \Delta}
    \end{prooftree}
  \end{align*}
  \begin{center}\vspace*{-2em}
  \end{center}
  \caption{Additional base rules}
  \label{fig:Additional}
\end{figure}

\noindent The following rule is very deceptive---even though it seems correct, its unsoundness has been formally proven in the theorem prover \textit{Isabelle/HOL}.
\begin{figure}
\begin{center}
\begin{prooftree}
  \hypo{\xi \s \Gamma \vdash \Psi_1 \chop \Psi_3, \Delta}
  \hypo{\xi \s \Gamma \vdash \Psi_2 \chop \Psi_3, \Delta}
  \infer[left label=\texttt{CH-$\land$R}]2{\xi \s \Gamma \vdash (\Psi_1 \land \Psi_2) \chop \Psi_3, \Delta}
\end{prooftree} 
  \caption{Unsound Rule}
  \label{fig:Unsound}
\end{center}
\end{figure}
\vspace{-6mm}
\begin{theorem}
The given rule in \Cref{fig:Unsound} is unsound.
\end{theorem}

\begin{proof}
  Consider recursive variables $(X_i)_{0 \leq i \leq 3}$ and $\Gamma := X_0$, $\Psi_1 := X_1$, $\Psi_2 := X_2$ and $\Psi_3 := X_3$. Assume the following valuation $\mathbb{V}$ with
  \begin{align*}
  &\mathbb{V}(X_0) := \{[\sigma_1, \sigma_2, \sigma_3, \sigma_4]\} &&
  \mathbb{V}(X_1) :=\{[\sigma_1], [\sigma_1, \sigma_2]\} \\
  &\mathbb{V}(X_2) := \{[\sigma_1, \sigma_2], [\sigma_1, \sigma_2, \sigma_3]\}
  &&\mathbb{V}(X_3) := \{[\sigma_1, \sigma_2, \sigma_3, \sigma_4], [\sigma_3, \sigma_4]\}
  \end{align*}
  where $(\sigma_i)_{1 \leq i \leq 4}$ are distinct states. This instantiation satisfies both premises, but not the conclusion of the rule, as $\llbracket X_1 \cap X_2\rrbracket_{\mathbb{V}} = \{[\sigma_1, \sigma_2]\}$, which results in the empty trace set when chopped together with $\mathbb{V}(X_3)$.
\end{proof}

\subsection{Alternative Fixed Point Induction Rule}
\label{sec:altern-fixp-induct}

\begin{figure}[h]
  \begin{center}
    \begin{align*} 
      &\begin{prooftree}
        \hypo{\xi \s P_\Gamma \vdash I}
        \hypo{\xi, (X|_{I}, \Psi) \s I, \Phi \vdash \Psi}
        \infer[left label=\texttt{FPI-ALT}]2{\xi \s \Gamma, \mu X. \Phi \vdash \Psi, \Delta}
      \end{prooftree}
    \end{align*}
  \end{center}
  \begin{center}\vspace*{-2em}
  \end{center}
  \caption{Alternative fixed point induction rule}
  \label{fig:fpadd}
\end{figure}

This rule requires $\xi$ to accept not only recursion variables, but
arbitrary fixed point formulas as its third triple composite. This
makes the calculus even more general, covering a wider range of
derivable sequents. A exemplary sequent that is derivable with
\texttt{FPI-ALT}, but \textit{not} with \texttt{FPI} could be
$$
P_\Gamma^2, \Phi_m \vdash true \chop x = 1
$$

\subsection{Theorems Needed in the Proof of Theorem~\ref{th:s1}}

\begin{theorem}[Partial Distributivity of Chop]
  \label{assocDC}
  For any trace formulas $\Phi_1, \Phi_2$ and $\Phi_3$ and any
  valuation $\mathbb{V}$, it holds~that
  $$\llbracket(\Phi_1 \lor \Phi_2) \chop \Phi_3\rrbracket_{\mathbb{V}}
  = \llbracket\Phi_1 \chop \Phi_3 \lor \Phi_2 \chop
  \Phi_3\rrbracket_{\mathbb{V}}$$ \text{\noindent and
  }
  $$\llbracket(\Phi_1 \land \Phi_2) \chop
  \Phi_3\rrbracket_{\mathbb{V}} \subseteq \llbracket\Phi_1 \chop \Phi_3 \land
  \Phi_2 \chop \Phi_3\rrbracket_{\mathbb{V}}$$
\end{theorem}

\begin{proof}
  Let us assume trace formulas $\Phi_1, \Phi_2$ and $\Phi_3$ are
  arbitrary, but fixed. Let us also assume valuation $\mathbb{V}$ is
  arbitrary, but fixed. Then also

  \begin{align*}
    \begin{split}
      & \llbracket(\Phi_1 \lor \Phi_2) \chop \Phi_3\rrbracket_{\mathbb{V}} \\
      & = \{ \sigma \cdot s \cdot \sigma' \text{ | } \sigma \cdot s \in \llbracket\Phi_1\rrbracket_{\mathbb{V}} \cup \llbracket\Phi_2\rrbracket_{\mathbb{V}} \land s \cdot \sigma' \in \llbracket \Phi_3\rrbracket_{\mathbb{V}} \} \\
      & = \{ \sigma \cdot s \cdot \sigma' \text{ | } \sigma \cdot s \in \llbracket\Phi_1\rrbracket_{\mathbb{V}} \land s \cdot \sigma' \in \llbracket \Phi_3\rrbracket_{\mathbb{V}} \} \\
      & \hspace{5mm} \cup \{ \sigma \cdot s \cdot \sigma' \text{ | } \sigma \cdot s \in \llbracket\Phi_2\rrbracket_{\mathbb{V}} \land s \cdot \sigma' \in \llbracket \Phi_3\rrbracket_{\mathbb{V}} \} \\
      & = \llbracket\Phi_1 \chop \Phi_3 \lor \Phi_2 \chop \Phi_3\rrbracket_{\mathbb{V}}
    \end{split}
  \end{align*}

  \begin{align*}
    \begin{split}
      & \llbracket(\Phi_1 \land \Phi_2) \chop \Phi_3\rrbracket_{\mathbb{V}} \\
      & = \{ \sigma \cdot s \cdot \sigma' \text{ | } \sigma \cdot s \in \llbracket\Phi_1\rrbracket_{\mathbb{V}} \cap \llbracket\Phi_2\rrbracket_{\mathbb{V}} \land s \cdot \sigma' \in \llbracket \Phi_3\rrbracket_{\mathbb{V}} \} \\
      & \subseteq \{ \sigma \cdot s \cdot \sigma' \text{ | } \sigma \cdot s \in \llbracket\Phi_1\rrbracket_{\mathbb{V}} \land s \cdot \sigma' \in \llbracket \Phi_3\rrbracket_{\mathbb{V}} \}  \\
      & \hspace{5mm} \cap \{ \sigma \cdot s \cdot \sigma' \text{ | } \sigma \cdot s \in \llbracket\Phi_2\rrbracket_{\mathbb{V}} \land s \cdot \sigma' \in \llbracket \Phi_3\rrbracket_{\mathbb{V}} \} \\
      & = \llbracket\Phi_1 \chop \Phi_3 \land \Phi_2 \chop \Phi_3\rrbracket_{\mathbb{V}}
    \end{split}
  \end{align*}
\end{proof}

\begin{theorem}[Equivalence of Repetitions inside Fixed Point]
  \label{repeat}
  For every recursion variable $X$, trace formula $\Phi$, valuation
  $\mathbb{V}$ and every positive natural number $n \geq 1$, it holds
  that
  $\llbracket\mu X. \Phi\rrbracket_{\mathbb{V}} = \llbracket\mu
  X. repeat_n(\Phi)\rrbracket_{\mathbb{V}}$.
\end{theorem}

\begin{proof}
  Let recursion variable $X$, trace formula $\Phi$, valuation
  $\mathbb{V}$ and $n \geq 1$ be arbitrary, but fixed. We define the
  following $\gamma$-sequences:
  $$
  (\gamma_1^i, \gamma_2^i)_{i \geq 0} \text{ s.t. } (\gamma_1^0, \gamma_2^0) = (\varnothing, \varnothing) \land \gamma_1^{i+1} = \llbracket\Phi\rrbracket_{\mathbb{V}[X \mapsto \gamma_1^i]} \land \gamma_2^{i+1} = \llbracket repeat_n(\Phi)\rrbracket_{\mathbb{V}[X \mapsto \gamma_2^i]}
  $$

  We will now prove $\gamma_1^{n*i} = \gamma_2^{i}$ for every
  $i \geq 0$ via natural induction over $i$. First, let $i = 0$. Then
  trivially $\gamma_1^0 = \varnothing = \gamma_2^0$. For the induction
  step, we assume $\gamma_1^{n*i} = \gamma_2^{i}$ for some fixed
  $i \geq 0$. Then
  \begin{align*}
    \begin{split}
      &\gamma_1^{n*(i+1)} = \gamma_1^{n*i+n} = \llbracket\Phi\rrbracket_{\mathbb{V}[X \mapsto \gamma_1^{n*i+(n-1)}]} = \llbracket\Phi\rrbracket_{\mathbb{V}[X \mapsto \llbracket\Phi\rrbracket_{\dots\mathbb{V}[X \mapsto \gamma_1^{n*i}]}]}
      \\ &= \llbracket repeat_n(\Phi)\rrbracket_{\mathbb{V}[X \mapsto \gamma_1^{n*i}]} = \llbracket repeat_n(\Phi)\rrbracket_{\mathbb{V}[X \mapsto \gamma_2^{i}]} = \gamma_2^{i+1}
    \end{split}
  \end{align*}
  Due to this result and the monotonicity of the function, we know
  that both sequences must, after possibly infinitely many steps, at
  some point have reached their least fixed points. Hence,
  $\llbracket\mu X. \Phi\rrbracket_{\mathbb{V}} = \llbracket\mu
  X. repeat_n(\Phi)\rrbracket_{\mathbb{V}}$, which is what needed to
  be shown in the first~place.
\end{proof}

\subsection{Proof of \Cref{th:s1} (Soundness of the Base Calculus)}
\label{sec:soundness-base}

\begin{proof}
  To prove that only valid sequents are derivable, we establish that
  all calculus rules are \textit{locally sound}. A calculus rule is
  called \textit{locally sound} if the conclusion is a valid sequent
  assuming all premises are valid sequents.  \smallskip

  \par{(\texttt{CASE}).} Let us assume
  $\llbracket \bigwedge \Gamma \land p\rrbracket_{\mathbb{V}}
  \subseteq \llbracket \bigvee \Delta\rrbracket_{\mathbb{V}} \text{
    and } \llbracket \bigwedge \Gamma \land
  \overline{p}\rrbracket_{\mathbb{V}} \subseteq \llbracket \bigvee
  \Delta\rrbracket_{\mathbb{V}}$.  To prove
  $\llbracket \bigwedge \Gamma\rrbracket_{\mathbb{V}} \subseteq
  \llbracket \bigvee \Delta\rrbracket_{\mathbb{V}}$, we perform a case
  distinction over predicate $p$. If we assume that $p$ is satisfied
  in the antecedent, then the first premise trivially concludes the
  case. In the case that the complement $\overline{p}$ is satisfied in
  the antecedent, the second premise trivially infers
  the~conclusion. \vspace{1em}

  \par{(\texttt{PRED}).} Let us assume
  $\llbracket \bigwedge P_\Gamma\rrbracket_{\mathbb{V}} \subseteq
  \llbracket q\rrbracket_{\mathbb{V}} \text{ and } \llbracket
  \bigwedge \Gamma \land q\rrbracket_{\mathbb{V}} \subseteq \llbracket
  \bigvee \Delta\rrbracket_{\mathbb{V}}$. Then also
  $\llbracket \bigwedge \Gamma\rrbracket_{\mathbb{V}} = \llbracket
  \bigwedge \Gamma \land \bigwedge P_\Gamma\rrbracket_{\mathbb{V}}
  \subseteq \llbracket \bigwedge \Gamma \land q\rrbracket_{\mathbb{V}}
  \subseteq \llbracket\bigvee
  \Delta\rrbracket_{\mathbb{V}}$. \vspace{1em}

  \par{(\texttt{CH-PREDL}).} Let us assume
  $\llbracket \bigwedge \Gamma \land p \land p \chop
  \Phi\rrbracket_{\mathbb{V}} \subseteq \llbracket\bigvee
  \Delta\rrbracket_{\mathbb{V}}$. Then also
  \begin{align*}
    \begin{split}
      & \llbracket \bigwedge \Gamma \land p \chop \Phi\rrbracket_{\mathbb{V}} = \llbracket \bigwedge \Gamma\rrbracket_{\mathbb{V}} \cap \{ \sigma \cdot s \cdot \sigma' \text{ | } \sigma \cdot s \models p \land s \cdot \sigma' \in \llbracket\Phi\rrbracket_{\mathbb{V}}\} \\
      & = \llbracket \bigwedge \Gamma\rrbracket_{\mathbb{V}} \cap \{s' \cdot \sigma \text{ | } s' \models p\} \cap \{ \sigma \cdot s \cdot \sigma' \text{ | } \sigma \cdot s \models p \land s \cdot \sigma' \in \llbracket\Phi\rrbracket_{\mathbb{V}}\} \\
      & = \llbracket \bigwedge \Gamma \land p \land p \chop \Phi\rrbracket_{\mathbb{V}} \subseteq \llbracket\bigvee \Delta\rrbracket_{\mathbb{V}}
    \end{split}
  \end{align*} 

  \par{(\texttt{CH-PREDR}).} Let us assume
  $\llbracket \bigwedge \Gamma \land q\rrbracket_{\mathbb{V}}
  \subseteq \llbracket true \chop \Psi \lor \bigvee
  \Delta\rrbracket_{\mathbb{V}}$. Then also
  \begin{align*}
    \begin{split}
      & \llbracket \bigwedge \Gamma \land q\rrbracket_{\mathbb{V}} = \llbracket \bigwedge \Gamma \land q\rrbracket_{\mathbb{V}} \cap \llbracket q\rrbracket_{\mathbb{V}} \subseteq  (\llbracket true \chop \Psi\rrbracket_{\mathbb{V}} \cup \llbracket\bigvee \Delta\rrbracket_{\mathbb{V}}) \cap \llbracket q\rrbracket_{\mathbb{V}} \\
      & \subseteq (\llbracket true \chop \Psi\rrbracket_{\mathbb{V}} \cap \llbracket q\rrbracket_{\mathbb{V}}) \cup \llbracket\bigvee \Delta\rrbracket_{\mathbb{V}} \\
      & = (\{s \cdot \sigma \text{ | } s \cdot \sigma \in \llbracket true \chop \Psi\rrbracket_{\mathbb{V}}\} \cap \{s \cdot \sigma \text{ | } s \models q\}) \cup \llbracket\bigvee \Delta\rrbracket_{\mathbb{V}} \\
      & \subseteq \llbracket q \chop \Psi \lor \bigvee \Delta\rrbracket_{\mathbb{V}}
    \end{split}
  \end{align*} 

  \par{(\texttt{REL}).} Let us assume the side condition
  $R|_{P(\Gamma)} \subseteq R'$ holds, implying that
  $\llbracket R\rrbracket_\mathbb{V} \cap \llbracket \bigwedge
  P_\Gamma\rrbracket_{\mathbb{V}} \subseteq \llbracket R'
  \rrbracket_\mathbb{V}$. Based on this, we
  conclude
  \begin{align*}\llbracket \bigwedge \Gamma \land
    R\rrbracket_\mathbb{V} \subseteq \llbracket R\rrbracket_\mathbb{V}
    \cap \llbracket \bigwedge P(\Gamma)\rrbracket_{\mathbb{V}}
    \subseteq \llbracket R' \rrbracket_\mathbb{V} \subseteq \llbracket
    R' \lor \bigvee \Delta\rrbracket_\mathbb{V}
  \end{align*}

  \par{(\texttt{RVAR}).} Let $\xi$ be arbitrary, but fixed, such that
  $(X_1|_{I}, X_2) \in \xi$. As such,
  $\llbracket X_1 \land I\rrbracket_{\mathbb{V}} \subseteq \llbracket
  X_2\rrbracket_{\mathbb{V}}$. Let us assume
  $\llbracket \bigwedge P_\Gamma\rrbracket_{\mathbb{V}} \subseteq
  \llbracket I\rrbracket_{\mathbb{V}}$.  Then
  also
  \begin{align*}\llbracket \bigwedge \Gamma \land
    X_1\rrbracket_\mathbb{V} \subseteq \llbracket \bigwedge P_\Gamma
    \land X_1\rrbracket_\mathbb{V} \subseteq \llbracket I \land
    X_1\rrbracket_\mathbb{V} \subseteq \llbracket
    X_2\rrbracket_\mathbb{V} \subseteq \llbracket X_2 \lor \bigvee
    \Delta\rrbracket_\mathbb{V}
  \end{align*}

  \par{(\texttt{CH-ID}).} Let us assume
  $\llbracket \bigwedge P_\Gamma \land Id\rrbracket_\mathbb{V}
  \subseteq \llbracket\Psi_i\rrbracket_\mathbb{V}$ for all $i$ with
  $1 \leq i \leq n$ and
  $\llbracket\bigwedge P_\Gamma \land \Phi_2\rrbracket_\mathbb{V}
  \subseteq \llbracket \bigvee_{1 \leq i \leq n}
  \Psi_i'\rrbracket_\mathbb{V}$. We trivially know that for any
  $(s, s') \in Id$, it must hold that $s = s'$. As such,
  \begin{align*}
    \begin{split}
      & \llbracket\bigwedge \Gamma \land Id \chop \Phi_2\rrbracket_\mathbb{V} \subseteq \llbracket\bigwedge P_\Gamma\rrbracket_\mathbb{V} \cap (\llbracket Id\rrbracket_\mathbb{V} \chop \llbracket\Phi_2\rrbracket_\mathbb{V}) \\
      & = \{ s \cdot s' \cdot \sigma \hspace{1mm} | \hspace{1mm} s \cdot s' \in \llbracket \bigwedge P_\Gamma \land Id\rrbracket_\mathbb{V} \land s' \cdot \sigma \in \llbracket\Phi_2\rrbracket_\mathbb{V} \} 
      \\
      & = \{ s \cdot s' \cdot \sigma \hspace{1mm} | \hspace{1mm} s \cdot s' \in \llbracket \bigwedge P_\Gamma \land Id\rrbracket_\mathbb{V} \land s' \cdot \sigma \in \llbracket\bigwedge P_\Gamma \land \Phi_2\rrbracket_\mathbb{V} \} \\
      & 
        \subseteq \bigcap_{1 \leq i \leq n} \{ s \cdot s' \cdot \sigma \hspace{1mm} | \hspace{1mm} s \cdot s' \in \llbracket\Psi_i\rrbracket_\mathbb{V} \land s' \cdot \sigma \in \llbracket \bigvee_{1 \leq i \leq n} \Psi_i'\rrbracket_\mathbb{V} \} \\
      &  
        = \bigcup_{1 \leq i \leq n} \{ s \cdot s' \cdot \sigma \hspace{1mm} | \hspace{1mm} s \cdot s' \in \llbracket\bigwedge_{1 \leq i \leq n} \Psi_i\rrbracket_\mathbb{V} \land s' \cdot \sigma \in \llbracket \Psi_i'\rrbracket_\mathbb{V} \} \\
      &
\subseteq \bigcup_{1 \leq i \leq n}\{ s \cdot s' \cdot \sigma \hspace{1mm} | \hspace{1mm} s \cdot s' \in \llbracket\Psi_i\rrbracket_\mathbb{V} \land s' \cdot \sigma \in \llbracket\Psi_i'\rrbracket_\mathbb{V} \} 
        \subseteq \llbracket \bigvee_{1 \leq i \leq n} \Psi_i \chop \Psi_i' \lor \bigvee \Delta\rrbracket_\mathbb{V}
    \end{split}
  \end{align*} 

  \par{(\texttt{CH-UPD}).} Let us assume
  $\llbracket \bigwedge P_\Gamma \land Sb_x^a\rrbracket_\mathbb{V}
  \subseteq \llbracket\Psi_i\rrbracket_\mathbb{V}$ for all i with
  $1 \leq i \leq n$ and
  $\llbracket\bigwedge spc_{x := a}(P_\Gamma) \land
  \Phi_2\rrbracket_\mathbb{V} \subseteq \llbracket\bigvee_{1 \leq i
    \leq n} \Psi_i'\rrbracket_\mathbb{V}$. We know that for any
  $(s, s') \in Sb_x^a$ with $s \models P_\Gamma$ for some predicate
  set $P_\Gamma$, it is guaranteed that
  $s' \models spc_{x := a}(P_\Gamma)$, which is based on the principle
  of strongest postconditions \cite{strongpost}. As such,
  \begin{align*}
    \begin{split}
      & \llbracket\bigwedge \Gamma \land Sb_x^a \chop \Phi_2\rrbracket_\mathbb{V} \subseteq \llbracket\bigwedge P_\Gamma\rrbracket_\mathbb{V} \cap (\llbracket Sb_x^a\rrbracket_\mathbb{V} \chop \llbracket\Phi_2\rrbracket_\mathbb{V}) \\
      & = \{ s \cdot s' \cdot \sigma \hspace{1mm} | \hspace{1mm} s \cdot s' \in \llbracket \bigwedge P_\Gamma \land Sb_x^a\rrbracket_\mathbb{V} \land s' \cdot \sigma \in \llbracket\Phi_2\rrbracket_\mathbb{V} \} \\
      & = \{ s \cdot s' \cdot \sigma \hspace{1mm} | \hspace{1mm} s \cdot s' \in \llbracket \bigwedge P_\Gamma \land Sb_x^a\rrbracket_\mathbb{V} \land s' \cdot \sigma \in \llbracket\bigwedge spc_{x := a}(P_\Gamma) \land \Phi_2\rrbracket_\mathbb{V} \} \\
      & \subseteq \bigcap_{1 \leq i \leq n} \{ s \cdot s' \cdot \sigma \hspace{1mm} | \hspace{1mm} s \cdot s' \in \llbracket\Psi_i\rrbracket_\mathbb{V} \land s' \cdot \sigma \in \llbracket \bigvee_{1 \leq i \leq n} \Psi_i'\rrbracket_\mathbb{V} \} \\
      & 
        = \bigcup_{1 \leq i \leq n} \{ s \cdot s' \cdot \sigma \hspace{1mm} | \hspace{1mm} s \cdot s' \in \llbracket\bigwedge_{1 \leq i \leq n} \Psi_i\rrbracket_\mathbb{V} \land s' \cdot \sigma \in \llbracket \Psi_i'\rrbracket_\mathbb{V} \} \\
      &
\subseteq \bigcup_{1 \leq i \leq n}\{ s \cdot s' \cdot \sigma \hspace{1mm} | \hspace{1mm} s \cdot s' \in \llbracket\Psi_i\rrbracket_\mathbb{V} \land s' \cdot \sigma \in \llbracket\Psi_i'\rrbracket_\mathbb{V} \} 
        \subseteq \llbracket\bigvee_{1 \leq i \leq n} \Psi_i \chop \Psi_i' \lor \bigvee \Delta\rrbracket_\mathbb{V}
    \end{split}
  \end{align*} 

  \par{(\texttt{END-ID}).} Let us assume
  $\llbracket \bigwedge P_\Gamma \land Id\rrbracket_\mathbb{V}
  \subseteq \llbracket\Psi_1\rrbracket_\mathbb{V} \text{ and }
  \llbracket\bigwedge P_\Gamma\rrbracket_\mathbb{V} \subseteq
  \llbracket \Psi_2\rrbracket_\mathbb{V}$. We trivially know that for
  any $(s, s') \in Id$, it must hold that $s = s'$. As such,
  \begin{align*}
    \begin{split}
      & \llbracket\bigwedge \Gamma \land Id\rrbracket_\mathbb{V} \subseteq \llbracket\bigwedge P_\Gamma \land Id\rrbracket_\mathbb{V} = \{ s \cdot s' \text{ | } s \cdot s' \in \llbracket Id\rrbracket_{\mathbb{V}} \land s \models \bigwedge P_\Gamma\} \\
      & = \{ s \cdot s' \text{ | } s \cdot s' \in \llbracket Id\rrbracket_{\mathbb{V}} \land s \models \bigwedge P_\Gamma \land s' \models \bigwedge P_\Gamma\} \\
      & \subseteq \{ s \cdot s' \text{ | } s \cdot s' \in \llbracket \Psi_1\rrbracket_\mathbb{V} \land s' \models \bigwedge P_\Gamma\} \subseteq \llbracket \Psi_1 \chop \Psi_2 \rrbracket_\mathbb{V}
    \end{split}
  \end{align*}

  \par{(\texttt{END-UPD}).} We assume
  $\llbracket \bigwedge P_\Gamma \land Sb_x^a\rrbracket_\mathbb{V}
  \subseteq \llbracket\Psi_1\rrbracket_\mathbb{V} \text{ and }
  \llbracket\bigwedge spc_{x := a}(P_\Gamma)\rrbracket_\mathbb{V}
  \subseteq \llbracket \Psi_2\rrbracket_\mathbb{V}$. We know that for
  any $(s, s') \in Sb_x^a$ with $s \models P_\Gamma$ for some
  predicate set $P_\Gamma$, it is guaranteed that
  $s' \models spc_{x := a}(P_\Gamma)$, which is based on the principle
  of strongest postconditions \cite{strongpost}. As such,
  \begin{align*}
    \begin{split}
      & \llbracket\bigwedge \Gamma \land Sb_x^a\rrbracket_\mathbb{V} \subseteq \llbracket\bigwedge P_\Gamma \land Sb_x^a\rrbracket_\mathbb{V} = \{ s \cdot s' \text{ | } s \cdot s' \in \llbracket Sb_x^a\rrbracket_{\mathbb{V}} \land s \models \bigwedge P_\Gamma\} \\
      & = \{ s \cdot s' \text{ | } s \cdot s' \in \llbracket Sb_x^a\rrbracket_{\mathbb{V}} \land s \models \bigwedge P_\Gamma \land s' \models \bigwedge spc_{x := a}(P_\Gamma)\} \\
      & \subseteq \{ s \cdot s' \text{ | } s \cdot s' \in \llbracket \Psi_1\rrbracket_\mathbb{V} \land s' \models \bigwedge spc_{x := a}(P_\Gamma)\} \subseteq \llbracket \Psi_1 \chop \Psi_2 \rrbracket_\mathbb{V}
    \end{split}
  \end{align*}

  \par{(\texttt{CH-$\lor$L}).} Assume
  $\llbracket \bigwedge \Gamma \land \Phi_1 \chop
  \Phi_3\rrbracket_\mathbb{V} \subseteq \llbracket\bigvee
  \Delta\rrbracket_\mathbb{V}$ and
  $\llbracket \bigwedge \Gamma \land \Phi_2 \chop
  \Phi_3\rrbracket_\mathbb{V} \subseteq \llbracket\bigvee
  \Delta\rrbracket_\mathbb{V}$. Using \Cref{assocDC} where marked with
  $^*$, we then infer
  \begin{align*}
    \begin{split}
      & \llbracket
      \bigwedge \Gamma \land (\Phi_1 \lor \Phi_2) \chop
      \Phi_3\rrbracket_\mathbb{V} \stackrel{*}{=} \llbracket \bigwedge
      \Gamma \land (\Phi_1 \chop \Phi_3) \lor (\Phi_2 \chop
        \Phi_3))\rrbracket_\mathbb{V} \\
      & = \llbracket \bigwedge
      \Gamma\rrbracket_\mathbb{V} \cap (\llbracket(\Phi_1 \chop
      \Phi_3)\rrbracket_\mathbb{V} \cup \llbracket(\Phi_2 \chop
        \Phi_3)\rrbracket_\mathbb{V}) \\
      & = (\llbracket \bigwedge
      \Gamma\rrbracket_\mathbb{V} \cap \llbracket\Phi_1 \chop
      \Phi_3\rrbracket_\mathbb{V}) \cup (\llbracket \bigwedge
      \Gamma\rrbracket_\mathbb{V} \cap \llbracket \Phi_2 \chop
        \Phi_3\rrbracket_\mathbb{V}) \\
      & = (\llbracket \bigwedge \Gamma
      \land (\Phi_1 \chop \Phi_3)\rrbracket_\mathbb{V}) \cup
      (\llbracket \bigwedge \Gamma \land (\Phi_2 \chop
      \Phi_3)\rrbracket_\mathbb{V}) \subseteq \llbracket\bigvee
        \Delta\rrbracket_\mathbb{V}
    \end{split}
  \end{align*}

  \par{(\texttt{CH-$\land$L}).} Let us assume
  $\llbracket \bigwedge \Gamma \land \Phi_1 \chop \Phi_3 \land \Phi_2
  \chop \Phi_3\rrbracket_\mathbb{V} \subseteq \llbracket\bigvee
  \Delta\rrbracket_\mathbb{V}$. Using \Cref{assocDC}, we then
  infer
  \begin{align*}\llbracket \bigwedge \Gamma \land (\Phi_1 \land
    \Phi_2) \chop \Phi_3\rrbracket_\mathbb{V} \stackrel{*}{\subseteq}
    \llbracket \bigwedge \Gamma \land \Phi_1 \chop \Phi_3 \land \Phi_2
    \chop \Phi_3\rrbracket_\mathbb{V} \subseteq \llbracket\bigvee
    \Delta\rrbracket_\mathbb{V}
  \end{align*}

  \par{(\texttt{CH-$\lor$R}).} Let us assume
  $\llbracket \bigwedge \Gamma\rrbracket_\mathbb{V} \subseteq
  \llbracket(\Psi_1 \chop \Psi_3) \lor (\Psi_2 \chop \Psi_3) \lor
  \bigvee \Delta\rrbracket_\mathbb{V}$. Using \Cref{assocDC}, we then
  infer
  \begin{align*}\llbracket \bigwedge \Gamma\rrbracket_\mathbb{V}
    \subseteq \llbracket(\Psi_1 \chop \Psi_3) \lor (\Psi_2 \chop
    \Psi_3) \lor \bigvee \Delta\rrbracket_\mathbb{V} \stackrel{*}{=}
    \llbracket(\Psi_1 \lor \Psi_2)\chop \Psi_3 \lor \bigvee
    \Delta\rrbracket_\mathbb{V}
  \end{align*}

  \par{(\texttt{ARB1}).} Let us assume
  $\llbracket \bigwedge \Gamma\rrbracket_{\mathbb{V}} \subseteq
  \llbracket \Psi \lor \bigvee \Delta\rrbracket_{\mathbb{V}}$. We then
  conclude
  \begin{align*}
    \begin{split}
      & \llbracket \bigwedge \Gamma\rrbracket_{\mathbb{V}} \subseteq \llbracket \Psi \lor \bigvee \Delta\rrbracket_{\mathbb{V}} = \{ s \cdot \sigma' \text{ | } s \cdot \sigma' \in \llbracket \Psi\rrbracket_{\mathbb{V}}\} \cup \llbracket\bigvee \Delta\rrbracket_{\mathbb{V}} \\
      & \subseteq \{ \sigma \cdot s \cdot \sigma' \text{ | } \sigma \cdot s \in \llbracket true\rrbracket_{\mathbb{V}} \land s \cdot \sigma' \in \llbracket \Psi\rrbracket_{\mathbb{V}}\} \cup \llbracket\bigvee \Delta\rrbracket_{\mathbb{V}} \\
      & = \llbracket true \chop \Psi \lor \bigvee \Delta\rrbracket_{\mathbb{V}}
    \end{split}
  \end{align*}

  \par{(\texttt{ARB2}).} Let us assume
  $\llbracket \bigwedge \Gamma \land \Phi_1 \chop
  \Phi_2\rrbracket_{\mathbb{V}} \subseteq \llbracket \Phi_1 \chop true
  \chop \Psi \lor \bigvee \Delta\rrbracket_{\mathbb{V}}$. Then
  \begin{align*}
    \begin{split}
      & \llbracket \bigwedge \Gamma \land \Phi_1 \chop \Phi_2\rrbracket_{\mathbb{V}} \subseteq \llbracket \Phi_1 \chop true \chop \Psi \lor \bigvee \Delta\rrbracket_{\mathbb{V}} \\
      & = \{ \sigma \cdot s \cdot \sigma' \text{ | } \sigma \cdot s \in \llbracket \Phi_1\rrbracket_{\mathbb{V}} \land s \cdot \sigma' \in \llbracket true \chop \Psi\rrbracket_{\mathbb{V}}\} \cup \llbracket\bigvee \Delta\rrbracket_{\mathbb{V}} \\
      & \subseteq \{ \sigma \cdot s \cdot \sigma' \text{ | } \sigma \cdot s \in \llbracket true\rrbracket_{\mathbb{V}} \land s \cdot \sigma' \in \llbracket true \chop \Psi\rrbracket_{\mathbb{V}}\} \cup \llbracket\bigvee \Delta\rrbracket_{\mathbb{V}} \\
      & = \{ \sigma \cdot s \cdot \sigma' \text{ | } \sigma \cdot s \in \llbracket true\rrbracket_{\mathbb{V}} \land s \cdot \sigma' \in \llbracket \Psi\rrbracket_{\mathbb{V}}\} \cup \llbracket\bigvee \Delta\rrbracket_{\mathbb{V}} \\
      & = \llbracket true \chop \Psi \lor \bigvee \Delta\rrbracket_{\mathbb{V}}
    \end{split}
  \end{align*}

  \par{(\texttt{UNFL}).} Let us assume
  $\llbracket\bigwedge \Gamma \land \Phi[\mu
  X. \Phi/X]\rrbracket_\mathbb{V} \subseteq \llbracket\bigvee
  \Delta\rrbracket_\mathbb{V}$. Due to fixed point unfolding, we
  trivially also know that
  $\llbracket\Phi[\mu X. \Phi/X]\rrbracket_\mathbb{V} = \llbracket\mu
  X. \Phi\rrbracket_\mathbb{V}$. As such,
  $\llbracket\bigwedge \Gamma \land \mu X. \Phi\rrbracket_\mathbb{V} =
  \llbracket\bigwedge \Gamma \land \Phi[\mu
  X. \Phi/X]\rrbracket_\mathbb{V} \subseteq \llbracket\bigvee
  \Delta\rrbracket_\mathbb{V}$. \vspace{1em}

  \par{(\texttt{UNFR}).} Let us assume
  $\llbracket\bigwedge \Gamma\rrbracket_\mathbb{V} \subseteq
  \llbracket\Psi[\mu X. \Psi/X] \lor \bigvee
  \Delta\rrbracket_\mathbb{V}$. Due to fixed point unfolding, we
  trivially also know that
  $\llbracket\Psi[\mu X. \Psi/X]\rrbracket_\mathbb{V} = \llbracket\mu
  X. \Psi\rrbracket_\mathbb{V}$. As such,
  $\llbracket\bigwedge \Gamma\rrbracket_\mathbb{V} \subseteq
  \llbracket\Psi[\mu X. \Psi/X] \lor \bigvee
  \Delta\rrbracket_\mathbb{V} = \llbracket\mu X. \Psi \lor\bigvee
  \Delta\rrbracket_\mathbb{V}$. \vspace{1em}

  \par{(\texttt{LENL}).} Let us assume
  $\llbracket \bigwedge \Gamma \land \mu
  X. repeat_i(\Phi)\rrbracket_{\mathbb{V}} \subseteq \llbracket
  \bigvee \Delta\rrbracket_{\mathbb{V}}$. Using \Cref{repeat} (marked
  with $^\dag$, we now conclude that
  $\llbracket \bigwedge \Gamma \land \mu
  X. \Phi\rrbracket_{\mathbb{V}} \stackrel{\dag}{=} \llbracket
  \bigwedge \Gamma \land \mu X. repeat_i(\Phi)\rrbracket_{\mathbb{V}}
  \subseteq \llbracket \bigvee \Delta\rrbracket_{\mathbb{V}}$ for all
  $i \geq 1$. \vspace{1em}

  \par{(\texttt{LENR}).} Let us assume
  $\llbracket \bigwedge \Gamma\rrbracket_{\mathbb{V}} \subseteq
  \llbracket \mu X. repeat_i(\Psi) \lor \bigvee
  \Delta\rrbracket_{\mathbb{V}}$. Using \Cref{repeat}, we now conclude
  that
  $\llbracket \bigwedge \Gamma\rrbracket_{\mathbb{V}} \subseteq
  \llbracket \mu X. repeat_i(\Psi) \lor \bigvee
  \Delta\rrbracket_{\mathbb{V}} \stackrel{\dag}{=} \llbracket \mu
  X. \Psi \lor \bigvee \Delta\rrbracket_{\mathbb{V}}$ for all
  $i \geq 1$. \vspace{1em}

  \par{(\texttt{CH-UNFL}).} Let us assume
  $\llbracket\bigwedge \Gamma \land (\Phi[\mu X. \Phi/X]) \chop
  \Phi'\rrbracket_\mathbb{V} \subseteq \llbracket\bigvee
  \Delta\rrbracket_\mathbb{V}$. Due to fixed point unfolding, we
  trivially also know that
  $\llbracket\Phi[\mu X. \Phi/X]\rrbracket_\mathbb{V} = \llbracket\mu
  X. \Phi\rrbracket_\mathbb{V}$. As such, we can also conclude
  that
  \begin{align*}
    \begin{split}
      & \llbracket\bigwedge \Gamma \land
        (\mu X. \Phi) \chop \Phi'\rrbracket_\mathbb{V} \\
      & =
      \llbracket\bigwedge \Gamma\rrbracket_\mathbb{V} \cap \{ \sigma
      \cdot s \cdot \sigma' \text{ | } \sigma \cdot s \in \llbracket
      \mu X. \Phi\rrbracket_{\mathbb{V}} \land s \cdot \sigma' \in
        \llbracket \Phi'\rrbracket_\mathbb{V}\} \\
      & =
      \llbracket\bigwedge \Gamma\rrbracket_\mathbb{V} \cap \{ \sigma
      \cdot s \cdot \sigma' \text{ | } \sigma \cdot s \in \llbracket
      \Phi[\mu X. \Phi/X]\rrbracket_{\mathbb{V}} \land s \cdot \sigma'
        \in \llbracket \Phi'\rrbracket_\mathbb{V}\} \\
      & =
      \llbracket\bigwedge \Gamma \land (\Phi[\mu X. \Phi/X]) \chop
      \Phi'\rrbracket_\mathbb{V} \subseteq \llbracket\bigvee
        \Delta\rrbracket_\mathbb{V}
    \end{split}
  \end{align*}

  \par{(\texttt{CH-UNFR}).} Let us assume
  $\llbracket\bigwedge \Gamma\rrbracket_\mathbb{V} \subseteq
  \llbracket(\Psi[\mu X. \Psi/X]) \chop \Psi' \lor \bigvee
  \Delta\rrbracket_\mathbb{V}$. Due to fixed point unfolding, we
  trivially also know that
  $\llbracket\Psi[\mu X. \Psi/X]\rrbracket_\mathbb{V} = \llbracket\mu
  X. \Psi\rrbracket_\mathbb{V}$. As such, we can also conclude
  that
  \begin{align*}
    \begin{split}
      & \llbracket\bigwedge
      \Gamma\rrbracket_\mathbb{V} \subseteq \llbracket(\Psi[\mu
      X. \Psi/X]) \chop \Psi' \lor \bigvee \Delta\rrbracket_\mathbb{V}
      \\
      & = \{ \sigma \cdot s \cdot \sigma' \text{ | } \sigma \cdot s
      \in \llbracket\Psi[\mu X. \Psi/X]\rrbracket_\mathbb{V} \land s
      \cdot \sigma' \in \llbracket \Psi'\rrbracket_\mathbb{V}\} \cup
        \llbracket\bigvee \Delta\rrbracket_\mathbb{V} \\
      & = \{ \sigma
      \cdot s \cdot \sigma' \text{ | } \sigma \cdot s \in
      \llbracket\mu X. \Psi\rrbracket_\mathbb{V} \land s \cdot \sigma'
      \in \llbracket \Psi'\rrbracket_\mathbb{V}\} \cup
        \llbracket\bigvee \Delta\rrbracket_\mathbb{V} \\
      & =
      \llbracket(\mu X. \Psi) \chop \Psi' \lor\bigvee
        \Delta\rrbracket_\mathbb{V}
    \end{split}
  \end{align*}

  \par{(\texttt{CH-LENL}).} Let us assume
  $\llbracket \bigwedge \Gamma \land (\mu X. repeat_i(\Phi)) \chop
  \Phi'\rrbracket_{\mathbb{V}} \subseteq \llbracket \bigvee
  \Delta\rrbracket_{\mathbb{V}}$. Using \Cref{repeat}, we can now
  conclude, that for any $i \geq 1$
  \begin{align*}
    \begin{split}
      & \llbracket\bigwedge \Gamma \land (\mu X. \Phi) \chop
        \Phi'\rrbracket_\mathbb{V} \\
      & = \llbracket\bigwedge
      \Gamma\rrbracket_\mathbb{V} \cap \{ \sigma \cdot s \cdot \sigma'
      \text{ | } \sigma \cdot s \in \llbracket \mu
      X. \Phi\rrbracket_{\mathbb{V}} \land s \cdot \sigma' \in
        \llbracket \Phi'\rrbracket_\mathbb{V}\} \\
      & \stackrel{\dag}{=}
      \llbracket\bigwedge \Gamma\rrbracket_\mathbb{V} \cap \{ \sigma
      \cdot s \cdot \sigma' \text{ | } \sigma \cdot s \in \llbracket
      \mu X. repeat_i(\Phi) \rrbracket_{\mathbb{V}} \land s \cdot
        \sigma' \in \llbracket \Phi'\rrbracket_\mathbb{V}\} \\
      & =
      \llbracket\bigwedge \Gamma \land (\mu X. repeat_i(\Phi)) \chop
      \Phi'\rrbracket_\mathbb{V} \subseteq \llbracket\bigvee
        \Delta\rrbracket_\mathbb{V}
    \end{split}
  \end{align*}

  \par{(\texttt{CH-LENR}).} Let us assume
  $\llbracket \bigwedge \Gamma\rrbracket_{\mathbb{V}} \subseteq
  \llbracket (\mu X. repeat_i(\Psi)) \chop \Psi' \lor \bigvee
  \Delta\rrbracket_{\mathbb{V}}$. Using \Cref{repeat}, we can now
  conclude, that for any $i \geq 1$
  \begin{align*}
    \begin{split}
      & \llbracket\bigwedge \Gamma\rrbracket_\mathbb{V} \subseteq
        \llbracket(\mu X. repeat_i(\Psi)) \chop \Psi' \lor \bigvee
        \Delta\rrbracket_\mathbb{V} \\
      & = \{ \sigma \cdot s \cdot
        \sigma' \text{ | } \sigma \cdot s \in \llbracket\mu
        X. repeat_i(\Psi)\rrbracket_\mathbb{V} \land s \cdot \sigma' \in
        \llbracket \Psi'\rrbracket_\mathbb{V}\} \cup \llbracket\bigvee
        \Delta\rrbracket_\mathbb{V} \\
      & \stackrel{\dag}{=} \{ \sigma \cdot
        s \cdot \sigma' \text{ | } \sigma \cdot s \in \llbracket\mu
        X. \Psi\rrbracket_\mathbb{V} \land s \cdot \sigma' \in
        \llbracket \Psi'\rrbracket_\mathbb{V}\} \cup \llbracket\bigvee
        \Delta\rrbracket_\mathbb{V} \\
      & = \llbracket(\mu X. \Psi) \chop
        \Psi' \lor\bigvee
        \Delta\rrbracket_\mathbb{V}
    \end{split}
  \end{align*}

  \par{(\texttt{FPI}).} Let us then assume the premises are valid,
  i.e.
  \begin{align*}
    \begin{split}
      & (1) \hspace{1mm} \llbracket P_\Gamma\rrbracket_{\mathbb{V}} \subseteq \llbracket I\rrbracket_{\mathbb{V}} \\
      & (2) \text{ If } \llbracket I \land X_1\rrbracket_{\mathbb{V}}\subseteq \llbracket X_2\rrbracket_{\mathbb{V}} \text{, then also } \llbracket I \land \Phi\rrbracket_\mathbb{V} \subseteq \llbracket \Psi\rrbracket_{\mathbb{V}}
    \end{split}
  \end{align*}

  Using the second premise, as well as \Cref{fpindproof}, we can now
  infer the proposition
  $\llbracket I \land \mu X_{1.}\Phi\rrbracket_\mathbb{V} \subseteq
  \llbracket \mu X_{2.} \Psi\rrbracket_\mathbb{V}$.  Hence, we
  conclude that
  \begin{align*}
    \llbracket \bigwedge \Gamma \land \mu X_{1.} \Phi\rrbracket_\mathbb{V} \subseteq \llbracket \bigwedge P_\Gamma \land \mu X_{1.} \Phi\rrbracket_\mathbb{V} \subseteq \llbracket I \land \mu X_{1.} \Phi\rrbracket_\mathbb{V} \subseteq \llbracket\mu X_{2.} \Psi \lor \bigvee \Delta\rrbracket_{\mathbb{V}}
  \end{align*}

  \iffalse
  \par{(\texttt{FPI-ALT}).} Let us assume the premises are valid, i.e.
  \begin{align*}
    \begin{split}
      & (1) \hspace{1mm} \llbracket P_\Gamma\rrbracket_{\mathbb{V}} \subseteq \llbracket I\rrbracket_{\mathbb{V}} \\

      & (2) \text{ If } \llbracket I \land X\rrbracket_{\mathbb{V}}\subseteq \llbracket \Psi\rrbracket_{\mathbb{V}} \text{, then also } \llbracket I \land \Phi\rrbracket_\mathbb{V} \subseteq \llbracket \Psi\rrbracket_{\mathbb{V}}
    \end{split}
   \end{align*}

   Using the second premise, as well as \Cref{fixconstr3}, we now can
   also infer the proposition
   $\llbracket I \land \mu X.\Phi\rrbracket_\mathbb{V} \subseteq
   \llbracket \Psi\rrbracket_\mathbb{V}$.  Hence, we conclude that
   \begin{align*}
     \llbracket \bigwedge \Gamma \land \mu X. \Phi\rrbracket_\mathbb{V} \subseteq \llbracket \bigwedge P_\Gamma \land \mu X. \Phi\rrbracket_\mathbb{V} \subseteq \llbracket I \land \mu X. \Phi\rrbracket_\mathbb{V} \subseteq \llbracket\Psi \lor \bigvee \Delta\rrbracket_{\mathbb{V}}
   \end{align*}
   \fi
\end{proof}

%%% Local Variables:
%%% mode: latex
%%% TeX-master: "main"
%%% End:

\section{Proofs of Contract Rules}
\label{app:contract}

\subsection{Additional Contract Application Rules}

\begin{figure}[h]
  \begin{center}
    \begin{align*} 
      & \begin{prooftree}
        \hypo{\mathbb{C}(m) = (pre, post)}
        \hypo{P_\Gamma \vdash_{\mathbb{C}} pre}
        \hypo{\xi \s P_\Gamma[v_{old}^i/v^i], post, \Phi \vdash_{\mathbb{C}} \Psi}
        \infer[left label=\texttt{CH-RVAR-EQ}]{1, 2}{\hspace{10mm} \xi \s \Gamma, X_m \chop \Phi \vdash_{\mathbb{C}} X_m \chop \Psi, \Delta \hspace{10mm}}
      \end{prooftree} \\[7pt]
      & \begin{prooftree}
        \hypo{\mathbb{C}(m) = (pre, post)}
        \hypo{P_\Gamma \vdash_{\mathbb{C}} I \land pre}
        \hypo{\hspace{-2mm}\xi, (X_m|_{I}, \Psi_1) \s I, \Phi_1 \vdash_{\mathbb{C}} \Psi_1}
        \hypo{\hspace{-2mm}\xi \s P_\Gamma[v_{old}^i/v^i], post, \Phi_2 \vdash_\mathbb{C} \Psi_2}
        \infer[left label=\texttt{CH-FPI-ALT}]{1, 3}{\hspace{28mm}\xi \s \Gamma, (\mu X_{m.} \Phi_1) \chop \Phi_2 \vdash_{\mathbb{C}} \Psi_1 \chop \Psi_2, \Delta \hspace{28mm}}
      \end{prooftree}
    \end{align*}
  \end{center}
  \begin{center}\vspace*{-2em}
  \end{center}
  \caption{Additional contract application rules}
  \label{fig:contractapply2}
\end{figure}

\subsection{Proof of \Cref{th:1}}

\begin{proof}
  Let recursion variable $X$, trace formula $\Phi$, valuation
  $\mathbb{V}$, and procedure contract $(pre, post)$ be arbitrary, but
  fixed. Let us assume the validity of $(pre, post)$ for $X$ in
  $\mathbb{V}$ implies its validity for $\Phi$ in $\mathbb{V}$. This
  is equivalent to saying that
  $\llbracket \langle pre(X)\rangle\rrbracket_{\mathbb{V}} \subseteq
  \llbracket \langle post(X)\rangle\rrbracket_{\mathbb{V}}$ implies
  that
  $\llbracket \langle pre(\Phi)\rangle\rrbracket_{\mathbb{V}}
  \subseteq \llbracket \langle
  post(\Phi)\rangle\rrbracket_{\mathbb{V}}$. Let
  $$
  P \equiv \bigwedge v_{old}^i = v^i \land pre
  $$
  be the predicates of the precondition encoding. Using the
  information contained in our premise, since $X$ specifies an
  arbitrary trace $\gamma$, we can also say that for any trace
  $\gamma$
  \begin{align*}
    \begin{split}
      & \llbracket P\rrbracket_{\mathbb{V}} \cap \gamma \chop State^+
        \subseteq \gamma \chop \llbracket post\rrbracket_{\mathbb{V}} \\
      & \text{ implies } \\
      & \llbracket P\rrbracket_{\mathbb{V}} \cap
      \llbracket \Phi\rrbracket_{\mathbb{V}[X \mapsto \gamma]} \chop
        State^+ \subseteq \llbracket \Phi\rrbracket_{\mathbb{V}[X
        \mapsto \gamma]} \chop \llbracket
        post\rrbracket_{\mathbb{V}}
    \end{split}
  \end{align*}
  We can now construct the following $\gamma$-sequence:
  $$
  (\gamma^i)_{i \geq 0} \text{ with } \gamma^0 = \varnothing \land \gamma^{i+1} = \llbracket \Phi\rrbracket_{\mathbb{V}[X \mapsto \gamma^i]}
  $$

  We prove via natural induction over $i$ that for every $\gamma^i$
  with $i \geq 0$:
  $\llbracket P\rrbracket_{\mathbb{V}} \cap \gamma^i \chop State^+
  \subseteq \gamma^i \chop \llbracket
  post\rrbracket_{\mathbb{V}}$. Let $i = 0$. Then
  trivially
  \begin{align*}
    \llbracket P\rrbracket_{\mathbb{V}} \cap
    \gamma^0 \chop State^+ = \llbracket P\rrbracket_{\mathbb{V}} \cap
    \varnothing \chop State^+ = \varnothing \subseteq \gamma^0 \chop
    \llbracket post\rrbracket_{\mathbb{V}}
  \end{align*}

  For the induction hypothesis, let $i \geq 0$ be fixed, such that it
  is guaranteed that
  $\llbracket P\rrbracket_{\mathbb{V}} \cap \gamma^i \chop State^+
  \subseteq \gamma^i \chop \llbracket post\rrbracket_{\mathbb{V}}$
  holds. Using our earlier premise, this is equivalent to saying that
  \begin{align*}
    \llbracket P\rrbracket_{\mathbb{V}} \cap \llbracket \Phi\rrbracket_{\mathbb{V}[X \mapsto \gamma^i]} \chop State^+ \subseteq \llbracket \Phi\rrbracket_{\mathbb{V}[X \mapsto \gamma^i]} \chop \llbracket post\rrbracket_{\mathbb{V}}
  \end{align*}

  Using this information, we can now complete the induction step by
  inferring that
  \begin{align*}
    \begin{split}
      & \llbracket P\rrbracket_{\mathbb{V}} \cap \gamma^{i+1} \chop State^+ = \llbracket P\rrbracket_{\mathbb{V}} \cap \llbracket \Phi\rrbracket_{\mathbb{V}[X \mapsto \gamma^i]} \chop State^+ \\
      & \subseteq \llbracket \Phi\rrbracket_{\mathbb{V}[X \mapsto \gamma^i]} \chop \llbracket post\rrbracket_{\mathbb{V}} = \gamma^{i+1} \chop \llbracket post\rrbracket_{\mathbb{V}}
    \end{split}
  \end{align*}

  Due to the monotonicity of the function, we know the
  $\gamma$-sequence above must, after possibly infinitely many steps,
  reach its least fixed point. Hence, we can conclude that also
  \begin{align*}
    \llbracket P\rrbracket_{\mathbb{V}} \cap \llbracket \mu X. \Phi\rrbracket_{\mathbb{V}} \chop State^+ \subseteq \llbracket \mu X. \Phi\rrbracket_{\mathbb{V}} \chop \llbracket post\rrbracket_{\mathbb{V}}
  \end{align*}

  This is again equivalent to
  $\llbracket \langle pre(\mu X. \Phi)\rangle\rrbracket_{\mathbb{V}}
  \subseteq \llbracket \langle post(\mu
  X. \Phi)\rangle\rrbracket_{\mathbb{V}}$, which needed to be shown in
  the first~place.
\end{proof}

\subsection{Application of Procedure Contracts}

\begin{lemma}[Application of Procedure Contracts]
  \label{mcappl}
  For any trace formulas $\Phi, \Psi$, recursion variable $X$,
  precondition $pre$, postcondition $post$, predicate $P$ and
  valuation $\mathbb{V}$, assuming procedure contract $(pre, post)$
  holds for $\Phi$ in $\mathbb{V}$, it must also hold that
  \begin{align*}
    \begin{split}
      & \{ \sigma \cdot s \cdot \sigma' \text{ \normalfont | } \sigma \cdot s \in \llbracket P \land pre \land \Phi\rrbracket_{\mathbb{V}} \land s \cdot \sigma' \in \llbracket \Psi \rrbracket_{\mathbb{V}}\} \\
      & \subseteq \{ \sigma \cdot s \cdot \sigma' \text{ \normalfont | } \sigma \cdot s \in \llbracket \Phi\rrbracket_\mathbb{V} \land s \cdot \sigma' \in \llbracket P[v_{old}^i/v^i] \land post \land \Psi\rrbracket_\mathbb{V}\}
    \end{split}
  \end{align*}
\end{lemma}

\begin{proof}
  Let us assume trace formulas $\Phi, \Psi$, recursion variable $X$,
  precondition $pre$, postcondition $post$, predicate $P$ and
  valuation $\mathbb{V}$ are arbitrary, but fixed, such that the
  procedure contract $(pre, post)$ holds for trace formula $\Phi$ in
  $\mathbb{V}$, i.e.
  $\llbracket \langle pre(\Phi)\rangle\rrbracket_{\mathbb{V}}
  \subseteq \llbracket \langle
  post(\Phi)\rangle\rrbracket_{\mathbb{V}}$. This encoding directly
  implies that
  \begin{align*}
    \llbracket \bigwedge v_{old}^i = v^i
    \land pre \land \Phi \chop true \rrbracket_{\mathbb{V}} \subseteq
    \llbracket \Phi \chop post \rrbracket_{\mathbb{V}}
  \end{align*}

  To infer the theorem, we first add the conjunctions
  $v_{old}^i = v^i$ to our left formula, which is allowed, as
  $v_{old}^i$ are assumed to be \textit{new program variables} not
  included in the formula yet.
  \begin{align*}
    \begin{split}
      & \{ \sigma \cdot s \cdot \sigma' \text{ | } \sigma \cdot s \in \llbracket P \land pre \land \Phi\rrbracket_{\mathbb{V}} \land s \cdot \sigma' \in \llbracket \Psi \rrbracket_{\mathbb{V}}\} \\
      & \subseteq \{ \sigma \cdot s \cdot \sigma' \text{ | } \sigma \cdot s \in \llbracket P \land \bigwedge v_{old}^i = v_i \land pre \land \Phi\rrbracket_{\mathbb{V}} \land s \cdot \sigma' \in \llbracket \Psi \rrbracket_{\mathbb{V}}\}
    \end{split}
  \end{align*}

  In the next step, we can then modify $\Phi$ to $\Phi \chop true$ in
  order to \textit{match} the formula with the encoding of the
  precondition $\langle pre(\Phi)\rangle$. Due to our matching
  encoding, we can then use the \textit{validity of the procedure
    contract}, as given in the premise, in order to add the encoding
  of the postcondition $\langle post(\Phi)\rangle$ to the
  formula. This is demonstrated as follows:
  \begin{align*}
    \begin{split}
      & \{ \sigma \cdot s \cdot \sigma' \text{ | } \sigma \cdot s \in \llbracket P \land \bigwedge v_{old}^i = v_i \land pre \land \Phi\rrbracket_{\mathbb{V}} \land s \cdot \sigma' \in \llbracket \Psi \rrbracket_{\mathbb{V}}\} \\
      & \subseteq \{ \sigma \cdot s \cdot \sigma' \text{ | } \sigma \cdot s \in \llbracket P \land \bigwedge v_{old}^i = v_i \land pre \land \Phi \chop true \rrbracket_{\mathbb{V}} \land s \cdot \sigma' \in \llbracket \Psi \rrbracket_{\mathbb{V}}\} \\
      & \subseteq \{ \sigma \cdot s \cdot \sigma' \text{ | } \sigma \cdot s \in \llbracket P \land \bigwedge v_{old}^i = v_i \land \Phi \chop post \rrbracket_{\mathbb{V}} \land s \cdot \sigma' \in \llbracket \Psi \rrbracket_{\mathbb{V}}\}
    \end{split}
  \end{align*}

  In the following step, we substitute every occurrence of $v^i$ in
  $P$ with $v_{old}^i$, which is possible, as we know that
  $v_{old}^i = v^i$ for all $i$. Considering that $post$ holds in the
  final state of $\sigma \cdot s$, we hence know that
  $s \models post$. As such, we can also add $post$ as a condition for
  the initial state of $s \cdot \sigma'$:
  \begin{align*}
    \begin{split}
      & \{ \sigma \cdot s \cdot \sigma' \text{ | } \sigma \cdot s \in \llbracket P \land \bigwedge v_{old}^i = v_i \land \Phi \chop post \rrbracket_{\mathbb{V}} \land s \cdot \sigma' \in \llbracket \Psi \rrbracket_{\mathbb{V}}\}\\
      & \subseteq \{ \sigma \cdot s \cdot \sigma' \text{ | } \sigma \cdot s \in \llbracket P[v_{old}^i/v^i] \land \Phi \chop post \rrbracket_{\mathbb{V}} \land s \cdot \sigma' \in \llbracket \Psi \rrbracket_{\mathbb{V}}\} \\
      & = \{ \sigma \cdot s \cdot \sigma' \text{ | } \sigma \cdot s \in \llbracket P[v_{old}^i/v^i] \land \Phi \rrbracket_{\mathbb{V}} \land s \cdot \sigma' \in \llbracket post \land \Psi \rrbracket_{\mathbb{V}}\}
    \end{split}
  \end{align*}

  Considering that $v_{old}^i$ are fresh program variables
  \underline{not} occurring in $\Phi$, we know that they stay
  unchanged during the execution of $\Phi$. Hence, all information
  about the old variables \underline{before} the execution of $\Phi$
  can simply be transferred intact until \underline{after} the
  execution of $\Phi$, which finally proves the lemma, as can be seen
  below:
  \begin{align*}
    \begin{split}
      & \{ \sigma \cdot s \cdot \sigma' \text{ | } \sigma \cdot s \in \llbracket P[v_{old}^i/v^i] \land \Phi \rrbracket_{\mathbb{V}} \land s \cdot \sigma' \in \llbracket post \land \Psi \rrbracket_{\mathbb{V}}\}\\
      & \subseteq \{ \sigma \cdot s \cdot \sigma' \text{ | } \sigma \cdot s \in \llbracket \Phi \rrbracket_{\mathbb{V}} \land s \cdot \sigma' \in \llbracket P[v_{old}^i/v^i] \land post \land \Psi \rrbracket_{\mathbb{V}}\}
    \end{split}
  \end{align*}
\end{proof}

\subsection{Proof of \Cref{th4}}

\begin{proof}
  We prove that each new rule is \textit{locally sound}. \vspace{1em}

  \par{(\texttt{MC}).} Let us assume the premises are valid,~i.e.
  \begin{align*}
    & (1) \hspace{1mm} \mathbb{C'} \text{ is valid in } \mathbb{V} \text{ implies } \llbracket \langle pre(\Phi)\rangle\rrbracket_{\mathbb{V}} \subseteq \llbracket \langle post(\Phi)\rangle\rrbracket_{\mathbb{V}} \\
    & (2) \hspace{1mm} \mathbb{C'} \text{ is valid in } \mathbb{V} \text{ implies } \llbracket \bigwedge \Gamma \land \mu X_{m.} \Phi\rrbracket_{\mathbb{V}} \subseteq \llbracket \bigvee \Delta\rrbracket_{\mathbb{V}}
  \end{align*}
  for $\mathbb{C'} = \mathbb{C}[m \mapsto (pre, post)]$ and
  $v_{old}^i \in fresh(Var)$.

  Using the first premise and \Cref{th:1}, we can infer that
  $(pre, post)$ is valid for $\mu X_{m.} \Phi$ in $\mathbb{V}$, i.e.
  \begin{align*}
    \llbracket \langle pre(\mu X_{m.} \Phi)\rangle\rrbracket_{\mathbb{V}} \subseteq \llbracket \langle post(\mu X_{m.} \Phi)\rangle\rrbracket_{\mathbb{V}}
  \end{align*}

  This only holds because $\mathbb{C'}$ being valid in $\mathbb{V}$ in
  this context means that $(pre, post)$ is valid for $X_m$, as no
  subformula $\mu X_{m.} \Psi$ can occur in $\Phi$. As such,
  $(pre, post)$ is valid for $\mu X_{m.} \Phi$ in $\mathbb{V}$. The
  second premise then tells us that
  \begin{align*}
    \llbracket \bigwedge \Gamma \land \mu X_{m.} \Phi\rrbracket_{\mathbb{V}} \subseteq \llbracket \bigvee \Delta\rrbracket_{\mathbb{V}}
  \end{align*}
  which is needed to be proven in the first place.
  \vspace{1em}

  \par{(\texttt{CH-MC}).} Let us assume the premises are valid,~i.e.
  \begin{align*}
    & (1) \hspace{1mm} \mathbb{C'} \text{ is valid in } \mathbb{V} \text{ implies } \llbracket \langle pre(\Phi_1)\rangle\rrbracket_{\mathbb{V}} \subseteq \llbracket \langle post(\Phi_1)\rangle\rrbracket_{\mathbb{V}} \\
    & (2) \hspace{1mm} \mathbb{C'} \text{ is valid in } \mathbb{V} \text{ implies } \llbracket \bigwedge \Gamma \land (\mu X_{m.} \Phi_1) \chop \Phi_2\rrbracket_{\mathbb{V}} \subseteq \llbracket \bigvee \Delta\rrbracket_{\mathbb{V}}
  \end{align*}
  for $\mathbb{C'} = \mathbb{C}[m \mapsto (pre, post)]$ and
  $v_{old}^i \in fresh(Var)$.

  Using the first premise and \Cref{th:1}, we can infer that
  $(pre, post)$ is valid for $\mu X_{m.} \Phi_1$ in $\mathbb{V}$, i.e.
  \begin{align*}
    \llbracket \langle pre(\mu X_{m.} \Phi_1)\rangle\rrbracket_{\mathbb{V}} \subseteq \llbracket \langle post(\mu X_{m.} \Phi_1)\rangle\rrbracket_{\mathbb{V}}
  \end{align*}

  This only holds because $\mathbb{C'}$ being valid in $\mathbb{V}$ in
  this context means that $(pre, post)$ is valid for $X_m$, as no
  subformula $\mu X_{m.} \Psi$ can occur in $\Phi_1$. As such,
  $(pre, post)$ is valid for $\mu X_{m.} \Phi_1$ in $\mathbb{V}$. The
  second premise then tells us that
  \begin{align*}
    \llbracket \bigwedge \Gamma \land (\mu X_{m.} \Phi_1) \chop \Phi_2\rrbracket_{\mathbb{V}} \subseteq \llbracket \bigvee \Delta\rrbracket_{\mathbb{V}}
  \end{align*}
  which needed to be proven. \vspace{1em}

  \par{(\texttt{CH-RVAR-EQ}).} Let us assume the premises are
  valid,~i.e.
  \begin{align*}
    & (1) \hspace{1mm} \llbracket P_\Gamma\rrbracket_{\mathbb{V}} \subseteq \llbracket pre\rrbracket_{\mathbb{V}} \\
    & (2) \hspace{1mm} \mathbb{C} \text{ being valid in } \mathbb{V} \text{ implies } \llbracket \bigwedge P_\Gamma[v_{old}^i/v^i] \land post \land \Phi\rrbracket_{\mathbb{V}} \subseteq \llbracket\Psi\rrbracket_{\mathbb{V}}
  \end{align*}
  for $\mathbb{C}(m) = (pre, post)$. Since
  $\mathbb{C}(m) = (pre, post)$, we can assume that $(pre, post)$
  holds for $X_m$ in $\mathbb{V}$. Hence, we can apply \Cref{mcappl}
  to conclude
  \begin{align*}
    \begin{split}
      & \llbracket \bigwedge \Gamma \land X_m \chop \Phi\rrbracket_\mathbb{V} \subseteq \llbracket \bigwedge P_\Gamma \land X_m \chop \Phi\rrbracket_\mathbb{V} \\
      & = \{ \sigma \cdot s \cdot \sigma' \text{ | } \sigma \cdot s \in \llbracket \bigwedge P_\Gamma \land X_m\rrbracket_{\mathbb{V}} \land s \cdot \sigma' \in \llbracket \Phi \rrbracket_{\mathbb{V}}\} \\
      & \subseteq \{ \sigma \cdot s \cdot \sigma' \text{ | } \sigma \cdot s \in \llbracket \bigwedge P_\Gamma \land pre \land X_m\rrbracket_{\mathbb{V}} \land s \cdot \sigma' \in \llbracket \Phi \rrbracket_{\mathbb{V}}\} \\
      & \subseteq \{ \sigma \cdot s \cdot \sigma' \text{ | } \sigma \cdot s \in \llbracket X_m\rrbracket_\mathbb{V} \land s \cdot \sigma' \in \llbracket \bigwedge P_\Gamma[v_{old}^i/v^i] \land post \land \Phi\rrbracket_\mathbb{V}\} \\
      & \subseteq \{ \sigma \cdot s \cdot \sigma' \text{ | } \sigma \cdot s \in \llbracket X_m\rrbracket_\mathbb{V} \land s \cdot \sigma' \in \llbracket \Psi\rrbracket_\mathbb{V}\} \subseteq \llbracket X_m \chop \Psi \lor \bigvee \Delta\rrbracket_\mathbb{V}
    \end{split}
  \end{align*}

  \par{(\texttt{CH-RVAR}).} Let $\xi$ be arbitrary, but fixed, such
  that $(X_m|_{p}, X) \in \xi$. As such,
  $\llbracket X_m \land p\rrbracket_{\mathbb{V}} \subseteq \llbracket
  X\rrbracket_{\mathbb{V}}$.  Let us assume the premises are
  valid,~i.e.
  \begin{align*}
    & (1) \hspace{1mm} \llbracket P_\Gamma\rrbracket_{\mathbb{V}} \subseteq \llbracket p \land pre\rrbracket_{\mathbb{V}} \\
    & (2) \hspace{1mm} \mathbb{C} \text{ being valid in } \mathbb{V} \text{ implies } \llbracket \bigwedge P_\Gamma[v_{old}^i/v^i] \land post \land \Phi\rrbracket_{\mathbb{V}} \subseteq \llbracket\Psi\rrbracket_{\mathbb{V}}
  \end{align*}
  for $\mathbb{C}(m) = (pre, post)$. Since
  $\mathbb{C}(m) = (pre, post)$, we can assume that $(pre, post)$
  holds for $X_m$ in $\mathbb{V}$. Hence, we can apply \Cref{mcappl}
  to conclude
  \begin{align*}
    \begin{split}
      & \llbracket \bigwedge \Gamma \land X_m \chop \Phi\rrbracket_\mathbb{V} \subseteq \llbracket \bigwedge P_\Gamma \land X_m \chop \Phi\rrbracket_\mathbb{V} \\
      & = \{ \sigma \cdot s \cdot \sigma' \text{ | } \sigma \cdot s \in \llbracket \bigwedge P_\Gamma \land X_m\rrbracket_{\mathbb{V}} \land s \cdot \sigma' \in \llbracket \Phi \rrbracket_{\mathbb{V}}\} \\
      & \subseteq \{ \sigma \cdot s \cdot \sigma' \text{ | } \sigma \cdot s \in \llbracket p \land \bigwedge P_\Gamma \land pre \land X_m\rrbracket_{\mathbb{V}} \land s \cdot \sigma' \in \llbracket \Phi \rrbracket_{\mathbb{V}}\} \\
      & \subseteq \{ \sigma \cdot s \cdot \sigma' \text{ | } \sigma \cdot s \in \llbracket p \land X_m\rrbracket_\mathbb{V} \land s \cdot \sigma' \in \llbracket \bigwedge P_\Gamma[v_{old}^i/v^i] \land post \land \Phi\rrbracket_\mathbb{V}\} \\
      & \subseteq \{ \sigma \cdot s \cdot \sigma' \text{ | } \sigma \cdot s \in \llbracket X\rrbracket_\mathbb{V} \land s \cdot \sigma' \in \llbracket \Psi\rrbracket_\mathbb{V}\}  \subseteq \llbracket X \chop \Psi \lor \bigvee \Delta\rrbracket_\mathbb{V}
    \end{split}
  \end{align*}

  \par{(\texttt{CH-FPI}).} Let us assume the premises are valid,~i.e.
  \begin{align*}
    & (1) \hspace{1mm} \llbracket P_\Gamma\rrbracket_{\mathbb{V}} \subseteq \llbracket I \land pre\rrbracket_{\mathbb{V}} \\
    & (2) \hspace{1mm} (X_m|_{I}, X) \in \xi \text{ implies } \llbracket I \land \Phi_1 \rrbracket_{\mathbb{V}} \subseteq \llbracket\Psi_1\rrbracket_{\mathbb{V}}. \\
    & (3) \hspace{1mm} \mathbb{C} \text{ being valid in } \mathbb{V} \text{ implies } \llbracket \bigwedge P_\Gamma[v_{old}^i/v^i] \land post \land \Phi_2\rrbracket_{\mathbb{V}} \subseteq \llbracket\Psi_2\rrbracket_{\mathbb{V}}
  \end{align*}
  for $\mathbb{C}(m) = (pre, post)$. Since
  $\mathbb{C}(m) = (pre, post)$, we can assume that $(pre, post)$
  holds for $\mu X_{m.} \Phi_1$ in $\mathbb{V}$. Using \Cref{mcappl},
  we conclude
  \begin{align*}
    \begin{split}
      & \llbracket \bigwedge \Gamma \land (\mu X_{m.} \Phi_1) \chop \Phi_2\rrbracket_\mathbb{V} \subseteq \llbracket \bigwedge P_\Gamma \land (\mu X_{m.} \Phi_1) \chop \Phi_2\rrbracket_\mathbb{V} \\
      & \subseteq \{ \sigma \cdot s \cdot \sigma' \text{ | } \sigma \cdot s \in \llbracket \bigwedge P_\Gamma \land \mu X_{m.} \Phi_1\rrbracket_\mathbb{V} \land s \cdot \sigma' \in \llbracket \Phi_2\rrbracket_\mathbb{V}\} \\
      & \subseteq \{ \sigma \cdot s \cdot \sigma' \text{ | } \sigma \cdot s \in \llbracket I \land \bigwedge P_\Gamma \land pre \land \mu X_{m.} \Phi_1\rrbracket_\mathbb{V} \land s \cdot \sigma' \in \llbracket \Phi_2\rrbracket_\mathbb{V}\} \\
      & \subseteq \{ \sigma \cdot s \cdot \sigma' \text{ | } \sigma \cdot s \in \llbracket I \land \mu X_{m.} \Phi_1\rrbracket_\mathbb{V} \land s \cdot \sigma' \in \llbracket P_\Gamma[v_{old}^i/v^i] \land post \land \Phi_2\rrbracket_\mathbb{V}\} \\
      & \subseteq \{ \sigma \cdot s \cdot \sigma' \text{ | } \sigma \cdot s \in \llbracket \mu X. \Psi_1\rrbracket_\mathbb{V} \land s \cdot \sigma' \in \llbracket \Psi_2\rrbracket_\mathbb{V}\} \subseteq \llbracket (\mu X. \Psi_1) \chop \Psi_2 \lor \bigvee \Delta\rrbracket
    \end{split}
  \end{align*}

  % Let us consider the following $\gamma$-sequence:
  % $$(\gamma^i)_{i \geq 0} \text{ with } \gamma^0 = \varnothing \land \gamma^{i+1} = \llbracket \Phi_1(X_1)\rrbracket_{\mathbb{V}[X_1 \mapsto \gamma^i]}$$
  % Since for any arbitrary $i \in \{1, ..., n\}$, it holds that $(pre^i, post^i)$ holds for $\mu X_{1.} \Phi_1$ - and considering that the $\gamma$-sequence will eventually (after possibly infinitely many steps) reach its least fixed point $\gamma^k = \llbracket \mu X_{1.} \Phi_1\rrbracket_{\mathbb{V}}$ - we know that $(pre^i, post^i)$ must hold for all $\bigcup_{i \geq 0} \gamma^i$. This implies that $(pre^i, post^i)$ holds for all $X_1$ in $\Phi_1(X_1)$ from bottom element $\gamma^0$ to its least fixed point $\gamma^k$. Combining this with the second premise and \textit{\color{blue}\Cref{fixconstr2}$^*$}, we can hence deduce that $\llbracket I \land \mu X_{1.}\Phi_1\rrbracket_\mathbb{V} \subseteq_{Var}^{{\color{blue}*}} \llbracket \mu X_{2.} \Psi_1\rrbracket_\mathbb{V}$.

  \par{(\texttt{CH-FPI-ALT}).} Let us assume the premises are
  valid,~i.e.
  \begin{align*}
    & (1) \hspace{1mm} \llbracket P_\Gamma\rrbracket_{\mathbb{V}} \subseteq \llbracket I \land pre\rrbracket_{\mathbb{V}} \\
    & (2) \hspace{1mm} (X_m|_{I}, \Psi_1) \in \xi \text{ implies } \llbracket I \land \Phi_1 \rrbracket_{\mathbb{V}} \subseteq \llbracket\Psi_1\rrbracket_{\mathbb{V}}. \\
    & (3) \hspace{1mm} \mathbb{C} \text{ being valid in } \mathbb{V} \text{ implies } \llbracket \bigwedge P_\Gamma[v_{old}^i/v^i] \land post \land \Phi_2\rrbracket_{\mathbb{V}} \subseteq \llbracket\Psi_2\rrbracket_{\mathbb{V}}
  \end{align*}
  for $\mathbb{C}(m) = (pre, post)$. Since
  $\mathbb{C}(m) = (pre, post)$, we can assume that $(pre, post)$
  holds for $\mu X_{m.} \Phi_1$ in $\mathbb{V}$. Using \Cref{mcappl},
  we conclude
  \begin{align*}
    \begin{split}
      & \llbracket \bigwedge \Gamma \land (\mu X_{m.} \Phi_1) \chop \Phi_2\rrbracket_\mathbb{V} \subseteq \llbracket \bigwedge P_\Gamma \land (\mu X_{m.} \Phi_1) \chop \Phi_2\rrbracket_\mathbb{V} \\
      & \subseteq \{ \sigma \cdot s \cdot \sigma' \text{ | } \sigma \cdot s \in \llbracket \bigwedge P_\Gamma \land \mu X_{m.} \Phi_1\rrbracket_\mathbb{V} \land s \cdot \sigma' \in \llbracket \Phi_2\rrbracket_\mathbb{V}\} \\
      & \subseteq \{ \sigma \cdot s \cdot \sigma' \text{ | } \sigma \cdot s \in \llbracket I \land \bigwedge P_\Gamma \land pre \land \mu X_{m.} \Phi_1\rrbracket_\mathbb{V} \land s \cdot \sigma' \in \llbracket \Phi_2\rrbracket_\mathbb{V}\} \\
      & \subseteq \{ \sigma \cdot s \cdot \sigma' \text{ | } \sigma \cdot s \in \llbracket I \land \mu X_{m.} \Phi_1\rrbracket_\mathbb{V} \land s \cdot \sigma' \in \llbracket P_\Gamma[v_{old}^i/v^i] \land post \land \Phi_2\rrbracket_\mathbb{V}\} \\
      & \subseteq \{ \sigma \cdot s \cdot \sigma' \text{ | } \sigma \cdot s \in \llbracket \Psi_1\rrbracket_\mathbb{V} \land s \cdot \sigma' \in \llbracket \Psi_2\rrbracket_\mathbb{V}\} \subseteq \llbracket \Psi_1 \chop \Psi_2 \lor \bigvee \Delta\rrbracket
    \end{split}
  \end{align*}
\end{proof}

%%% Local Variables:
%%% mode: latex
%%% TeX-master: "main"
%%% End:

\section{Proof of Synchronization Rule}
\label{app:sync}

% \setcounter{chapter}{3}
% \setcounter{section}{1}
% \setcounter{subsection}{0}
% \renewcommand{\thesection}{\Alph{chapter}}%

\iffalse
  
\subsection{Additional Synchronization Rule}

\begin{figure}[h]
  \begin{center}
    \begin{align*} 
      & \begin{prooftree}
        \hypo{\xi \s \Gamma \vdash (\mu X. \Psi_1') \chop \Psi_2, \Delta}
        \infer[left label=\texttt{CH-SYNC}]1[$L(gr(\Psi_1')) \subseteq L(gr(\Psi_1))$]{\xi \s \Gamma \vdash (\mu X. \Psi_1) \chop \Psi_2, \Delta}
      \end{prooftree}
    \end{align*}
  \end{center}
  \begin{center}\vspace*{-2em}
  \end{center}
  \caption{Additional synchronization Rule}
  \label{fig:addsync}
\end{figure}

\fi

\subsection{Additional Lemmas}

\begin{lemma}[Equivalence of Fixed Point Representations]
  \label{representation}
  For any fixed relation $R$, fixed recursion variable $X$, valuation
  $\mathbb{V}$ and chop formula $\Psi \in CF_{(R,X)}$ with
  $\Psi = \bigvee_{1 \leq j \leq n} \varphi_j$, let the following be a
  $\gamma$-sequence $(\gamma^i)_{i \geq 0}$ induced by fixed point
  operation $\mu X. \Psi$:
  $$
  (\gamma^i)_{i \geq 0} \text{ with } \gamma^0 = \varnothing \land \gamma^{i+1} = \llbracket\Psi\rrbracket_{\mathbb{V}[X \mapsto \gamma^i]}\enspace.
  $$

  Also let the following be a sequence of sets of primitive chop
  formulas $(C^i)_{i \geq 0}$ induced by chop formula $\Psi$:
  $$
  C^0 = \varnothing \text{ and } C^{i+1} = \bigcup_{1 \leq j \leq n} \{\varphi_j[c^1/X^{(1)}]\cdots[c^z/X^{(z)}]) \text{\normalfont{ | }} c^1, \ldots, c^z \in C^i\}
  $$
  where $X^{(i)}$ refers to the $i$-th occurrence of $X$ in a
  primitive chop formula $\varphi_j$. Then
  $\gamma^i = \llbracket C^i\rrbracket_{\mathbb{V}}$ for all
  $i \geq 0$.
\end{lemma}

\begin{proof}
  Let us assume relation $R$, recursion variable $X$, valuation
  $\mathbb{V}$ and chop formula
  $\Psi = \bigvee_{1 \leq j \leq n} \varphi_j \in CF_{(R, X)}$ are
  arbitrary, but fixed. We apply natural induction on $i \geq 0$ to
  prove that $\gamma^i = \llbracket C^i\rrbracket_{\mathbb{V}}$. For
  that purpose, we first establish that
  $\gamma^0 = \varnothing = \llbracket C^0\rrbracket_{\mathbb{V}}$.
  For the induction hypothesis, let us assume that
  $\gamma^i = \llbracket C^i\rrbracket_{\mathbb{V}}$ for a fixed
  $i \geq 0$. Then we can infer
  \begin{align*}
    \begin{split}
      & \gamma^{i+1} = \llbracket \bigvee_{1 \leq j \leq n} \varphi_j\rrbracket_{\mathbb{V}[X \mapsto \gamma^i]} = \bigcup_{1 \leq j \leq n} \llbracket\varphi_j\rrbracket_{\mathbb{V}[X \mapsto \gamma^i]} = \bigcup_{1 \leq j \leq n} \llbracket\varphi_j\rrbracket_{\mathbb{V}[X \mapsto \llbracket C^i\rrbracket_{\mathbb{V}}]} 
      \\
      & = \bigcup_{1 \leq j \leq n} \{\llbracket\varphi_j[c^1/X^{(1)}]\cdots[c^z/X^{(z)}]\rrbracket_{\mathbb{V}} \text{ | } c^1,\ldots, c^z \in C^i\} \\
      & = \llbracket \bigcup_{1 \leq j \leq n} \{\varphi_j[c^1/X^{(1)}]\cdots[c^z/X^{(z)}] \text{ | } c^1,\ldots, c^z \in C^i\}\rrbracket_{\mathbb{V}} = \llbracket C^{i+1}\rrbracket_{\mathbb{V}}
    \end{split}
  \end{align*}

  We have established that
  $\gamma^i = \llbracket C^i\rrbracket_{\mathbb{V}}$ holds for all
  $i \geq 0$.
\end{proof}

\begin{lemma}[Derivability of Primitive Chop Formulas in Grammar]
  \label{word}
  For any fixed relation $R$, fixed recursion variable $X$, chop
  formula $\Psi \in CF_{(R,X)}$, assuming the sequence of sets of
  primitive chop formulas $(C^i)_{i \geq 0}$ with
  $$
  C^0 = \varnothing\text{ and }C^{i+1} = \bigcup_{1 \leq j \leq n} \{\varphi_j[c^1/X^{(1)}]\cdots[c^z/X^{(z)}]) \text{\normalfont{ | }} c^1,\ldots, c^z \in C^i\}
  $$
  then also $\bigcup_{i \geq 0} grammarize(C^i) = L(gr(\Psi))$.
\end{lemma}

\begin{proof}
  Let us assume relation $R$, recursion variable $X$ and corresponding
  chop formula
  $\Psi = \bigvee_{1 \leq j \leq n} \varphi_j \in CF_{(R, X)}$ are
  arbitrary, but fixed. We now have to deduce that
  $\bigcup_{i \geq 0} grammarize(C^i) = L(gr(\Psi))$. We split the
  proof of the equality into a forward- and backward-direction.

  $\Rightarrow:$ First show that
  $\bigcup_{i \geq 0} grammarize(C^i) \subseteq L(gr(\Psi))$ via
  induction over $i$. The induction base
  $$grammarize(C^0) = \varnothing \subseteq L(gr(\Psi))$$ trivially
  holds.
  % Let us thus assume $w_1 \in grammarize(C^1)$ is arbitrary, but
  % fixed. Then there must exist a $c^1 \in C^1$ with
  % $grammarize(c^1) = w_1$. This $c_1$, by construction, must match
  % $\varphi_j$ for some $j$. We can then derive $w_1$ with the
  % derivation rule $X \rightarrow \gamma$ with
  % $\gamma = grammarize(\varphi_j)$.
  For the induction step, we fix $i$ and assume, as the induction
  hypothesis, that $grammarize(C^i)$ can be derived in $gr(\Psi)$. We
  will now show that the words in $grammarize(C^{i+1})$ can also be
  derived in $gr(\Psi)$. Let us assume
  $w_{i+1} \in grammarize(C^{i+1})$ is arbitrary, but fixed. Then
  there exists a $c^{i+1} \in C^{i+1}$ with
  $grammarize(c^{i+1}) = w_{i+1}$. This means that there exists a
  $\varphi_j$ for some $j$ and $c^1,\ldots, c^z \in C^i$, such that
  $c^{i+1} = \varphi_j[c^1/X^{(1)}]\cdots[c^z/X^{(z)}]$. We can now
  derive $w_{i+1}$ by applying the derivation rule
  $X \rightarrow \gamma$ with $\gamma = grammarize(\varphi_j)$, where
  each occurrence $X^{(m)}$ inside $\gamma$ is again derived by
  applying the derivation of $grammarize(c^m)$. This derivation must
  already exist, because $c^m \in C^i$, and as such
  $grammarize(c^m) \in grammarize(C^i)$, which lies in the domain of
  our induction hypothesis.

  $\Leftarrow:$ We need to prove that
  $L(gr(\Psi)) \subseteq \bigcup_{i \geq 0} grammarize(C^i)$. To this
  end, let $w \in L(gr(\Psi))$ be arbitrary, but fixed, and have a
  derivation depth $k$. We prove via induction over derivation depth
  $k$, that also $w \in grammarize(C^k)$.  Let us first assume that
  $w$ has depth $1$. Then there exists a derivation rule
  $X \rightarrow grammarize(\varphi_j) \text{ for some } j \text{ with
  } 1 \leq j \leq n$, such that $grammarize(\varphi_j) = w$. Since
  $\varphi_j$ can only contain relation $R$, this also implies that
  $\varphi_j \in C^1$, hence $w \in grammarize(C^1)$.

  For the induction step, we fix $k$ and assume, as the induction
  hypothesis, that any word with a derivation depth of $k$ is included
  in $grammarize(C^k)$. Then, let us assume that word $w$ has depth
  $k+1$. Hence, there must exist a derivation rule
  $X \rightarrow grammarize(\varphi_j) \text{ for some } j \text{ with
  } 1 \leq j \leq n$, ensuring that any derivation of its internal
  non-terminals $X$ must have a derivation depth of $k$, such that
  word $w$ can be derived. Using the induction hypothesis, we hence
  know that all derived words $w^1,\ldots, w^z$ of the internal
  non-terminals $X$ must be included in $grammarize(C^k)$. Since
  $C^{k+1}$ includes all $\varphi_j$, where all occurrences of its
  recursion variables $X$ have been replaced by elements of $C^k$, our
  word $w$ must also be included in $grammarize(C^{k+1})$, i.e.
  $w \in grammarize(C^{k+1})$.\end{proof}

\begin{lemma}[Fixed Point Trace Representation in Language]
  \label{gr1}
  For any fixed relation $R$, recursion variable $X$, valuation
  $\mathbb{V}$, chop formula $\Psi \in CF_{(R,X)}$, let
  $$
  c_{\sigma} := \overbrace{R \chop \ldots \chop R}^{l-times}
  $$ be a
  primitive chop formula of length $l$. Then the following two
  statements must hold at the same time:
  \begin{enumerate}
  \item There exists a trace
    $\sigma \in \llbracket\mu X. \Psi\rrbracket_{\mathbb{V}}$ of
    length $l \geq 1$ with
    $\llbracket c_{\sigma}\rrbracket_{\mathbb{V}} =
    \{\sigma\}$. 
  \item $grammarize(c_{\sigma}) \in L(gr(\Psi))$.
  \end{enumerate}
\end{lemma}

\begin{figure}
\centering
\captionsetup{justification=centering}
\begin{tikzpicture}[shorten >=1pt,node distance=3.5cm, on grid,auto] 
   \node[rectangle, draw] (q_3) {$\llbracket \mu X. \Psi\rrbracket_{\mathbb{V}}$}; 
   \node[rectangle, draw] (q_0) [right=of q_3] {$(\gamma^i)_{i \geq 0}$}; 
   \node[rectangle, draw] (q_1) [right=of q_0] {$(C^i)_{i \geq 0}$}; 
   \node[rectangle, draw] (q_2) [right=of q_1] {$L(gr(\Psi))$}; 
    \path[densely dotted, <->] 
    (q_3) edge  node {} (q_0)
    (q_0) edge  node {\Cref{representation}} (q_1)
    (q_1) edge  node {\Cref{word}} (q_2);
\end{tikzpicture}
\caption{Visualization of established proof connections}
\label{fig1}
\end{figure}

\begin{proof}
  Let us assume relation $R$, recursion variable $X$, valuation
  $\mathbb{V}$ and chop formula
  $\Psi = \bigvee_{1 \leq j \leq n} \varphi_j \in CF_{(R, X)}$ are
  arbitrary, but fixed. Let $c_{\sigma}$ be a primitive chop formula
  of length $l$. We now prove the lemma by establishing the forward-
  and backward-direction, which both follow the outline visualized in
  \Cref{fig1}.

  $\Rightarrow$: Let us assume trace
  $\sigma \in \llbracket\mu X. \Psi\rrbracket_{\mathbb{V}}$ of length
  $l \geq 1$ is arbitrary, but fixed, such that
  $\llbracket c_{\sigma}\rrbracket_{\mathbb{V}} = \{\sigma\}$. We now
  consider the following $\gamma$-sequence:
  $$
  (\gamma^i)_{i \geq 0} \text{ with } \gamma^0 = \varnothing \land \gamma^{i+1} = \llbracket \Psi\rrbracket_{\mathbb{V}[X \mapsto \gamma^i]}
  $$

  This sequence must (after possibly infinitely many steps) have
  reached its least fixed point
  $\llbracket\mu X. \Psi\rrbracket_{\mathbb{V}}$. Since
  $\sigma \in \llbracket\mu X. \Psi\rrbracket_{\mathbb{V}}$ is a
  finite trace by default, there exists a $k \geq 0$ such that
  $\sigma \in \gamma^k$, i.e.\ $\sigma$ has been generated after $k$
  iterations. Using \Cref{representation}, we know that for the
  sequence of sets of primitive chop formulas $(C^i)_{i \geq 0}$ with
  $$
  C^0 = \varnothing \text{ and } C^{i+1} = \bigcup_{1 \leq j \leq n} \{\varphi_j[c^1/X^{(1)}]\cdots[c^z/X^{(z)}]) \text{ | } c^1,\ldots, c^z \in C^i\}
  $$
  it holds that $\gamma^i = \llbracket C^i\rrbracket_{\mathbb{V}}$ for
  all $i \geq 0$. Since $\sigma \in \gamma^k$, we thus know that
  $\sigma \in \llbracket C^k\rrbracket_{\mathbb{V}}$. Any primitive
  chop formula included in $C^k$ can only consist of relation $R$ as
  its atoms. This is trivial, as $C^0$ is the empty set, while
  $C^{i+1}$ replaces all occurrences of recursion variable~$X$ with
  primitive chop formulas of $C^{i}$. Since $\sigma$ is of length $l$,
  $c_{\sigma} \in C^k$ must hold as well.

  We can now construct a derivation for $grammarize(c_{\sigma})$ in
  $gr(\Psi)$. Since $c_{\sigma} \in C^k$,
  $grammarize(c_{\sigma}) \in grammarize(C^k)$ must also hold. Using
  \Cref{word}, this implies $grammarize(c_{\sigma}) \in L(gr(\Psi))$,
  which needed to be proven.

  $\Leftarrow$: Let us assume that
  $grammarize(c_{\sigma}) \in L(gr(\Psi))$. We consider the sequence
  of sets of primitive chop formulas $(C^i)_{i \geq 0}$ with
  $$
  C^0 = \varnothing \text{ and } C^{i+1} = \bigcup_{1 \leq j \leq n}
  \{\varphi_j[c^1/X^{(1)}]\cdots[c^z/X^{(z)}]) \text{ | } c^1, \ldots, c^z
  \in C^i\}
  $$

  Applying \Cref{word}, since
  $grammarize(c_{\sigma}) \in L(gr(\Psi))$, we know that there exists
  a corresponding set of primitive chop formulas $C^k$, such that
  necessarily $grammarize(c_{\sigma}) \in grammarize(C^k)$. This
  implies $c_{\sigma} \in C^k$ for some $k \geq 0$.  Since
  $c_{\sigma}$ is of length $l$, there exists some trace $\sigma$ of
  length $l$ with
  $\llbracket c_{\sigma}\rrbracket_{\mathbb{V}} = \{\sigma\}$. This
  implies that $\sigma \in \llbracket C_k\rrbracket_{\mathbb{V}}$. Let
  us consider the $\gamma$-sequence
  $$(\gamma^i)_{i \geq 0} \text{ with } \gamma^0 = \varnothing \land
  \gamma^{i+1} = \llbracket\Psi\rrbracket_{\mathbb{V}[X \mapsto
    \gamma^i]}$$ generated by the fixed point operation $\mu
  X. \Psi$. Due to \Cref{representation}, we also know that
  $\sigma \in \llbracket C^k\rrbracket_{\mathbb{V}}$ implies
  $\sigma \in \gamma^k$. $\sigma \in \gamma^k$ again implies that
  $\sigma \in \llbracket \mu X. \Psi\rrbracket_{\mathbb{V}}$, which
  needed to be proven.
\end{proof}

\begin{lemma}[Application of Trace Synchronization]
  \label{tslemma}
  For any fixed relation $R$, recursion variable $X$, valuation
  $\mathbb{V}$ and chop formulas $\Psi,\, \Psi' \in CF_{(R,X)}$, if we
  assume $L(gr(\Psi')) \subseteq L(gr(\Psi))$, then also
  $\llbracket\mu X. \Psi'\rrbracket_{\mathbb{V}} \subseteq
  \llbracket\mu X. \Psi\rrbracket_{\mathbb{V}}$.
\end{lemma}

\begin{proof}
  Let us assume relation $R$, recursion variable $X$, valuation
  $\mathbb{V}$ and chop formulas $\Psi, \Psi' \in CF_{(R,X)}$ are
  arbitrary, but fixed, such that
  $L(gr(\Psi')) \subseteq L(gr(\Psi))$. Let us choose a trace
  $\sigma \in \llbracket\mu X. \Psi'\rrbracket_{\mathbb{V}}$ of length
  $l \geq 1$ arbitrary, but fixed. Let us now consider the primitive
  chop formula $c_{\sigma}$ with
  $$c_{\sigma} = \overbrace{R \chop \ldots \chop R}^{l-times}$$
  Considering that $\Psi'$ is a chop formula, trace
  $\sigma \in \llbracket\mu X. \Psi'\rrbracket_{\mathbb{V}}$ of length
  $l$ must be a trace that has $R$ applied $l-times$ as a
  chop-sequence, i.e.\
  $\llbracket c_{\sigma}\rrbracket_{\mathbb{V}} = \{\sigma\}$. Using
  \Cref{gr1}, we hence know that
  $grammarize(c_{\sigma}) \in L(gr(\Psi'))$. Using our premise, we can
  deduce that $grammarize(c_{\sigma}) \in L(gr(\Psi))$. Applying
  \Cref{gr1} again, we can infer that there also exists a trace
  $\sigma' \in \llbracket\mu X. \Psi\rrbracket_{\mathbb{V}}$ with
  $\llbracket c_{\sigma}\rrbracket_{\mathbb{V}} = \{\sigma'\}$. Since
  $\{\sigma\}= \llbracket c_{\sigma}\rrbracket_{\mathbb{V}} =
  \{\sigma'\}$, we conclude that
  $\sigma \in \llbracket\mu X. \Psi\rrbracket_{\mathbb{V}}$, which
  was to be proven.
\end{proof}

\subsection{Proof of \Cref{t2}}

\begin{proof} We prove that each new rule is \textit{locally
    sound}. \vspace{1em}

  \par{(\texttt{SYNC}).} Let us assume
  $\llbracket \bigwedge \Gamma\rrbracket_{\mathbb{V}} \subseteq
  \llbracket \mu X. \Psi' \lor \bigvee \Delta\rrbracket_{\mathbb{V}}$.
  Let us further assume that the side condition holds,
  i.e. $L(gr(\Psi')) \subseteq L(gr(\Psi))$. Using \Cref{tslemma}, we
  infer that
  \begin{align*}
    \llbracket \bigwedge \Gamma\rrbracket_{\mathbb{V}} \subseteq \llbracket \mu X. \Psi' \lor \bigvee \Delta\rrbracket_{\mathbb{V}} \subseteq \llbracket \mu X. \Psi \lor \bigvee \Delta\rrbracket_{\mathbb{V}}
  \end{align*}

  \iffalse
    
  \par{(\texttt{CH-SYNC}).} Let us assume
  $\llbracket \bigwedge \Gamma\rrbracket_{\mathbb{V}} \subseteq
  \llbracket (\mu X. \Psi_1') \chop \Psi_2 \lor \bigvee
  \Delta\rrbracket_{\mathbb{V}}$.  Let us further assume the side
  condition, i.e. $L(gr(\Psi_1')) \subseteq L(gr(\Psi_1))$. Using
  \Cref{tslemma}, we infer
  \begin{align*}
    \begin{split}
      & \llbracket \bigwedge \Gamma\rrbracket_{\mathbb{V}} \subseteq  \llbracket (\mu X. \Psi_1') \chop \Psi_2 \lor \bigvee \Delta\rrbracket_{\mathbb{V}} \\
      & = \{ \sigma \cdot s \cdot \sigma' \text{ | } \sigma \cdot s \in \llbracket \mu X. \Psi_1'\rrbracket_{\mathbb{V}} \land s \cdot \sigma' \in \llbracket \Psi_2\rrbracket_{\mathbb{V}} \} \cup \llbracket \bigvee \Delta\rrbracket_{\mathbb{V}} \\
      & \subseteq \{ \sigma \cdot s \cdot \sigma' \text{ | } \sigma \cdot s \in \llbracket \mu X. \Psi_1\rrbracket_{\mathbb{V}} \land s \cdot \sigma' \in \llbracket \Psi_2\rrbracket_{\mathbb{V}} \} \cup \llbracket \bigvee \Delta\rrbracket_{\mathbb{V}} \\
      & = \llbracket (\mu X. \Psi_1) \chop \Psi_2 \lor \bigvee \Delta\rrbracket_{\mathbb{V}}
    \end{split}
  \end{align*}

  \fi

\end{proof}

%%% Local Variables:
%%% mode: latex
%%% TeX-master: "main"
%%% End:

\fi

\end{document}